\documentclass[11pt,fullpage]{article}
\usepackage{amsfonts}
\usepackage{amsmath}
\usepackage{amssymb}
\usepackage{amsthm}
\usepackage[dvips]{graphicx}
\usepackage{subfigure}
\usepackage{wrapfig}
\usepackage{floatflt}
\usepackage{geometry}
\usepackage{algorithm}
\usepackage{algorithmic}
\usepackage{color}
\usepackage{tweaklist}
\usepackage{mathrsfs}

\renewcommand{\itemhook}{\setlength{\topsep}{4pt}%
  \setlength{\itemsep}{0pt} }

\renewcommand{\enumhook}{\setlength{\topsep}{4pt}%
  \setlength{\itemsep}{0pt} }

\newtheoremstyle{spacesaver}%
        {7pt plus2pt minus4pt}% [space above] plus
        {7pt plus2pt minus4pt}% [space below]
        {\upshape}% body font
        {}%
        {\bfseries\upshape}% heading font
        {}%
        {1em}%
        {}%
\theoremstyle{spacesaver}
\newtheorem{thm}{Theorem}[section]
\newtheorem{definition}[thm]{Definition}
\newtheorem{defn}[thm]{Definition}

\newtheorem{reduc}[thm]{Reduction}
\newtheorem{obs}[thm]{Observation}

\newtheorem{coro}[thm]{Corollary}
\newtheorem{lem}[thm]{Lemma}

\newcommand{\Bnote}[1]{}

\newcommand{\nc}{\newcommand}

\newcommand{\bigo}{{\mathcal{O}}}

\newcommand{\supsrc}{{\bullet}}
\newcommand{\msg}{\mathfrak{M}}
\newcommand{\cf}{{f}}
\newcommand{\nwc}{{\Phi}}
\newcommand{\In}{{\mathsf{In}}}
\newcommand{\hgap}[1]{{\gamma(#1)}}
\newcommand{\gsdef}{{\mathsf{SD}}}
\newcommand{\sdef}{{\mathsf{LSD}}}

\newcommand{\bal}{{A}}
\newcommand{\balg}[1]{{\langle #1 \rangle}}
\newcommand{\spn}{+}
\newcommand{\intsct}{\cap}
\newcommand{\vctr}[1]{{\mathbf{#1}}}
\newcommand{\Sigmas}{{\{\Sigma_e\}}}
\newcommand{\fs}{{\{f_e\}}}
\newcommand{\atoms}{{\operatorname{At}}}
\newcommand{\gdiv}{{\Gamma}}
\newcommand{\ggdiv}{{\Upsilon}}
\nc{\bvec}{\vctr{e}}

\nc{\Salph}{\Sigma_{S}}
\nc{\Ealph}{\Sigma_e}
\nc{\Src}{S}
\nc{\Snk}{\mathcal{T}}
\nc{\Malph}{\Sigma_{\mathcal{M}}}
\nc{\networkcode}[1]{(G^{#1}, \msg, \{\Ealph^{#1}\}, \{f^{#1}_e\})}

\nc{\gf}{G \cdot F}
\nc{\vgf}{v^{\gf}}
\nc{\Agf}{A^{\gf}}
\nc{\Bgf}{B^{\gf}}
\nc{\Pgf}{P_{\gf}}

\nc{\Spn}{+}
\nc{\Intsct}{\cap}
\nc{\field}{{\mathbb{F}}}

\nc{\ce}{C\&E }
\nc{\word}{functions}
\nc{\src}[1]{(\bullet, #1)}
\nc{\st}{\text{ s.t. }}
\nc{\CallSer}{\textbf{LinSerialize}($\nwc$) }
\nc{\GenSer}{\textbf{GenSerialize}($\nwc$)}

\title{The Serializability of Network Codes}
\author{
Anna Blasiak \footnotemark[1]  \and
Robert Kleinberg \footnotemark[2] 
}
\footnotetext[1]{Department of Computer Science, Cornell University, Ithaca NY 14853.  E-mail: \texttt{ablasiak@cs.cornell.edu}.  Supported an by an NDSEG Graduate Fellowship, an AT\&T Labs Graduate Fellowship, and an NSF Graduate
Fellowship.}
\footnotetext[2]{Department of Computer Science, Cornell University, Ithaca NY 14853.  E-mail: \texttt{rdk@cs.cornell.edu}.  Supported by NSF grant CCF-0729102, a Microsoft Research New Faculty Fellowship, and an Alfred P.~Sloan Foundation Fellowship.}

\begin{document}
\date{}
\maketitle
\thispagestyle{empty}

\begin{abstract}
Network coding theory studies the transmission of information in networks whose vertices may
perform nontrivial encoding and decoding operations on data as it passes through the network.  The main approach to deciding the feasibility of network coding problems aims to reduce the problem to optimization over a polytope of ``entropic vectors" subject to constraints imposed by the network structure.  In the case of directed acyclic graphs, these constraints are completely understood, but for general graphs the problem of enumerating them  remains open: it is not known how to classify the  constraints implied by a property that we call \emph{serializability}, which refers to the absence of paradoxical circular dependencies in a network code.

In this work we initiate the first systematic study of the constraints imposed on a network
code by serializability.  We find that serializability cannot be detected solely by evaluating the Shannon entropy of edge sets in the graph, but nevertheless, we give a polynomial-time algorithm that decides the serializability of a network code.  We define a certificate of non-serializability, called an \emph{information vortex}, that plays a role in the theory of serializability comparable to the role of fractional cuts in multicommodity flow theory, including a type of min-max relation.  Finally, we study the \emph{serializability deficit} of a network code, defined as the 
minimum number of extra bits that must be sent in order to make it serializable.  For linear codes, we show that it is NP-hard to approximate this parameter within a constant factor, and we demonstrate some surprising facts about the behavior of this parameter under parallel composition of codes.

\end{abstract}

\newpage
\setcounter{page}{1}

\section{Introduction}
\label{sec:intro}
Network coding theory studies the transmission of information in networks whose vertices may
perform nontrivial encoding and decoding operations on data as it passes through the network.  More specifically, a network code consists of a network with specified sender and receiver edges and coding functions on each edge.  The classic definition of a network code requires that each vertex can compute the message on every outgoing edge from the messages received on its incoming edges, and that each receiver is sent the message it requires.  In directed acyclic graphs, a network code that satisfies these requirements specifies a valid communication protocol.  However, in graphs with cycles this need not be the case; the definition does not preclude the possibility of cyclic dependencies among coding functions.  Therefore, in graphs with cycles we also require that a network code is \emph{serializable}, meaning it correctly summarizes a communication protocol in which symbols are transmitted on edges over time, and each symbol transmitted by a vertex is computed without knowledge of information it will receive in the future. The present paper is devoted to the study of characterizing the constraint of serializability.

\paragraph{Motivation.}
The central question in the area of network coding is to determine the amount by which coding can increase the rate of information flow as compared to transferring information without coding.  Crucial to answering this question is developing tools to find upper bounds for the network coding rate.  The question of serializability must be considered in order to determine tight upper bounds on network codes in cyclic graphs.  Determining tight upper bounds is especially relevant to one of the most important open problems in network coding, the \emph{undirected $k$-pairs conjecture}, which states that in undirected graphs with $k$ sender-receiver pairs, coding cannot increase the maximum rate of information flow;  that is, the network coding rate is the same as the multicommodity flow rate.  Apart from its intrinsic interest, the conjecture also has important complexity-theoretic implications: for example, if true, it implies an affirmative answer to a 20-year-old conjecture regarding the I/O complexity of matrix transposition \cite{SODA}.

Almost all efforts to produce upper bounds on the network coding rate have focused on
the following construction.  We regard
each edge of the network as defining a
random variable on a probability space and then associate each set of edges with the Shannon entropy of the
joint distribution of their random variables.
This gives us a vector of non-negative
numbers, one for each edge set, called the
\emph{entropic vector} of the network code.
The closure of the set of entropic vectors of network codes forms a convex set, and network coding problems can be expressed as optimization problems over this set~\cite{YeungFC}.  In much previous work, tight upper bounds have been constructed by combining the constraints that define this convex set.  However, this technique is limited because there is no known description
of all these constraints.

There are two types of constraints: the purely
information-theoretic ones (i.e.,~those
that hold universally for all $n$-tuples of
random variables, regardless of their interpretation
as coding functions on edges of a network) and
the constraints derived from the combinatorial
structure of the network.  The former type of
constraints include the so-called 
Shannon and non-Shannon inequalities,
and are currently a topic of intense 
investigation~\cite{ChanYeung,DFZ06,Matus,MMR,ZY98}.
The latter type of constraints --- namely, those
determined by the network structure --- are trivial to characterize in the case of directed acyclic
graphs: if one imposes a constraint that the
entropy of each node's incoming edge set equals
the entropy of all of its incoming and outgoing 
edges, then these constraints together with
the purely information-theoretic ones imply
all other constraints resulting from the network
structure~\cite{YZ99}.  However, in graphs with cycles 
there are additional constraints determined from the network structure. 
% It is necessary to consider the constraints implied by serializability, because, as shown in \cite{CE}, we can create examples of coding functions that are not serializable, and ``achieve'' a faster coding rate than any serialziable code.  We need to exclude such codes in order to get a tight lower bound.

In a series of work on finding network coding upper bounds, large classes of information inequalities in graphs with cycles were discovered independently by Jain et al.~\cite{Jain}, Kramer and Savari~\cite{KramerSavari},
and Harvey et al.~\cite{TransIT}.  These go by the names
\emph{crypto inequality}, \emph{PdE bound}, and 
\emph{informational dominance bound}, respectively.
In various forms, all of them describe a situation
in which the information on one set of edges 
completely determines the information on another
set of edges.
%, in which case there is an information
%inequality involving those two components of the
%entropic vector.  
A more general necessary condition for serializability was presented in
a recent paper by Harvey et al.~\cite{CE};
% presents a six-term information inequality that follows from serializability (
we will henceforth refer to this information inequality as the \emph{Chicken and Egg} inequality; see Theorem~\ref{thm:2cycle}.  Though in all the previous work the inequalities used were sufficient to prove the needed bounds on the specific graphs analyzed in the paper, no one has asked if this set of inequalities provides a complete characterization of serializability.  This inspires the following natural questions: Are the information inequalities given in previous work sufficient to characterize serializability?  Does there exist a set of information theoretic inequalities that gives a sufficient condition for serializability?  Is there a finite set of information theoretic inequalities implied by serializability?  Is there any ``nice'' condition that is necessary and sufficient serializability?  We see these questions not only as interesting for a general understanding of network coding, but also a key ingredient to eventually developing algorithms and upper bounds for network coding in general graphs.
       
\paragraph{Our contributions.}
Our work is the first systematic study of
criteria for serializability of network
codes.  We find that serializability cannot be detected solely
from the entropic vector of the network
code; a counter-example is given in Section \ref{sec:2cycle}.  
This leads us to focus the paper on two independent, but dual, questions: Is there any efficiently verifiable necessary and sufficient condition for serializability?  What is the complete set of entropy inequalities implied by serializability?  

We answer the first question in the affirmative in Section \ref{sec:serialize} by providing an algorithm to decide whether a code is  serializable.
The running time of this algorithm is 
polynomial in the cardinalities of
the edge alphabets, and it is polynomial
in their dimensions in the case of linear
network codes.
We answer the second question for the 2-cycle in Section \ref{sec:2cycle}, giving four inequalities derived from the network structure, and showing that any entropic vector satisfying those inequalities as well as Shannon's inequalities can be realized by a serializable network code.  (Though structurally simple, the 2-cycle graph has been an important source of inspiration for information inequalities in prior work, including the crypto inequality~\cite{Jain}, the informational dominance bound~\cite{TransIT}, and the Chicken and Egg inequality~\cite{CE}.)  Disappointingly, we do not know if this result extends beyond the 2-cycle.

Beyond providing an algorithm for deciding if a network code is
serializable, our work provides important insights into the 
property of serializability.  
In Section~\ref{sec:serialize} we 
define a certificate that we call an \emph{information vortex} 
that is a necessary and sufficient condition for non-serializability.  
For linear network codes, an information vortex consists
of linear subspaces of the
dual of the message space.  For general network codes, it consists of Boolean
subalgebras of the power set of the message set.  We prove a number of
theorems about information vortices that
suggest their role in the theory of network coding
may be similar to the role of fractional cuts in
network flow theory.  In particular, we prove a type of
min-max relation between serializable codes 
and information vortices: under a suitable
definition of \emph{serializable restriction} it 
holds that every network code has a unique
maximal serializable restriction, a unique
minimal information vortex, and these two
objects coincide.

Finally, motivated by examples in which non-serializable codes, whose coding functions have dimension $n$, can be serialized by adding a single bit, we consider the idea of a network code being ``close'' to serializable.  We formalize this by studying a parameter we call the \emph{serializability
deficit} of a network code, defined as the 
minimum number of extra bits that must be sent 
in order to make it serializable.
For linear codes, we show that it is NP-hard to 
approximate this parameter within a constant factor.
We also demonstrate, perhaps surprisingly, that 
the serialization deficit may behave subadditively
under parallel composition: when executing two
independent copies of a network code, the 
serialization deficit may scale up by a factor
less than two.  In fact, for every $\delta>0$ 
there is a network code $\Phi$ and a positive
integer $n$ such that the serialization deficit
of $\Phi$ grows by a factor less than $\delta n$
when executing $n$ independent copies of $\Phi$.
Despite these examples, we are able to prove
that for any non-serializable linear code $\Phi$ there
is a constant $c_\Phi$ such that the
serialization deficit of $n$ independent
executions of $\Phi$ is at least $c_\Phi n$.  The concept of an information vortex is crucial to our results on serializability deficit.

\paragraph{Related work.}

For a general introduction to network coding we 
refer the reader to~\cite{ARL-thesis,YeungCZ-book}.
There is a standard definition of network codes
in directed acyclic graphs (Definition~\ref{def:nwc}
below) but in many papers on graphs with cycles the 
definition is either not explicit (e.g.~\cite{ACLY})
or is restricted to special classes of codes 
(e.g.~\cite{ErezFeder}).  Precise and general definitions of network codes in graphs
with cycles appear in~\cite{KoetterMedard,SODA,TransIT}
and the equivalence of these definitions (modulo some
differing assumptions about nodes' memory) is
proven in~\cite{ARL-thesis}.  The definition for serializability that we set forth in Section~\ref{sec:defs} 
was used, but never formally defined, in \cite{CE}.  In its essence it is the same as the ``graph over time'' definition given in \cite{ARL-thesis} but requires less cumbersome notation.

\section{Definitions}
\label{sec:defs}

We define a network code to operate on a directed multigraph we call a \emph{sourced graph}, denoted $G = (V,E,S). $\footnote{In prior work it is customary for the underlying network to also have a special set of receiving edges.  Specifying a special set of receivers is irrelevant in our work, so we omit them for convenience, but everything we do can be easily extended to include receivers.}   $S$ is a set of special edges, called \emph{sources} or \emph{source edges}, that have a head but no tail.  We denote
a source with head $s$ by an ordered pair
$(\supsrc,s)$.  Elements of $E \cup S$ are called \emph{edges} and elements of $E$ are called \emph{ordinary edges}.  For a vertex $v$, we let $\In(v) = \{(u,v) \in E\}$ be the set of edges whose head is $v$.  For an edge $e=(u,v) \in E$, we also use $\In(e) = \In(u)$ to denote  the set of \emph{incoming edges} to $e$.  % $\In(\src{s}) = \In(\bullet) =  \emptyset$.  

A network code in a sourced graph specifies a
protocol for communicating symbols on error-free
channels corresponding to the graph's ordinary
edges, given the tuple of messages that originate
at the source edges.  

\begin{definition}
\label{def:nwc}
A \emph{network code} is specified by a $4$-tuple
$\nwc = (G,\msg,\{\Sigma_e\}_{e \in E \cup S},\{\cf_e\}_{e \in E \cup S})$
where $G=(V,E,S)$ is a sourced graph, $\msg$ is a
set whose elements are called \emph{message-tuples},
and for all edges $e \in E \cup S,$
$\Sigma_e$ is a set called 
the \emph{alphabet} of $e$ and 
$\cf_e : \msg \rightarrow \Sigma_e$ is 
a function called the \emph{coding function}
of $e$.  If $e$ is an edge and $e_1,\ldots,e_k$
are the elements of $\In(e)$
then the value of the coding function $\cf_e$ must
be completely determined by the values of
$\cf_{e_1},\ldots,\cf_{e_k}$.  In other words, 
there must exist a function $g_e : \prod_{i=1}^k
\Sigma_{e_i} \rightarrow \Sigma_e$ such that
for all $m \in \msg$, 
$
\cf_e(m) = g_e(\cf_{e_1}(m),\ldots,\cf_{e_k}(m)).
$
\end{definition}
In graphs with cycles a code can have cyclic dependencies so Definition~\ref{def:nwc} does not suffice to characterize the notion of a valid network code.  We must impose a further constraint that we call \emph{serializability},
which requires that the network code summarizes
a complete execution of a communication protocol
in which every bit transmitted by a vertex depends
only on bits that it has already received.  

Below we define serializability formally using a definition implicit in \cite{CE}.
%The protocol itself can be represented by a DAG
%whose edges correspond to the transmissions.
%To express the property that a network code in 
%a graph with cycles summarizes a sequence of 
%transmissions in another graph, we define a 
%\emph{homomorphism} of network codes.  The following
%definitions apply to a pair of sourced graphs
%$G=(V,E,S)$ and $G'=(V',E',S')$ and a pair of
%network codes $\Phi=(G,\msg,\{\Sigma_e\},\{f_e\})$
%and $\Phi'=(G',\msg',\{\Sigma'_e\},\{f'_e\})$, where $\nwc$ will be the original network code and $\nwc'$ the network code in the DAG.
%Note that the composition of two weak homomorphisms is a weak homomorphism.  
%\begin{defn} \label{def:homomorphism}
%\label{1a} \label{1b}
%\label{1c} \label{2a} \label{2b} 
%A \emph{network code homomorphism} $h : \Phi' \rightarrow \Phi$ consists of:
%\begin{enumerate}
%\item A graph homomorphism\footnote{Note that, in contrast to the usual definition of graph homomorphism, we allow $h(u')=h(v')$ for an edge $e'=(u',v')$ even if the vertex $h(u')$ doesn't have a self-loop.} $h: G' \rightarrow G$ \st for every edge $e' =(u',v') \in E'$, either $h$ maps $u'$ and $v'$ to the same vertex $v \in V$, or to the endpoints of an edge $(u,v) \in E$; and $h$ restricts to a bijection between $S$ and $S'$.
%\item A bijection $h^* : \msg \rightarrow \msg'$ and functions $h_{e'}^* : \Sigma_{h(e')} \rightarrow \Sigma_{e'}$ $\forall\;e' \in E' \cup S'$ such that $h_{s'}^*$ is a bijection for all $s' \in S'$
%and $h_{e'}^*(f_{h(e')}(m)) = f_{e'}(h^*(m))$ for all
%$e' \in E' \cup S'$ such that $h(e') \not\in V$.  
%\end{enumerate}
%\end{defn}

\begin{definition} \label{def:ser} 
\label{2c}
A network code $\Phi$ is serializable if for all $e \in E$ there exists a set of alphabets $\Sigma_e^{(1...k)} = \left\{\Sigma_e^{(1)}, \Sigma_e^{(2)}, \ldots, \Sigma_e^{(k)}\right\}$ and a set of functions $f_e^{(1...k)} = \left\{f_e^{(1)}, f_e^{(2)}, \ldots, f_e^{(k)}\right\}$ such that 
\begin{enumerate}
\item \label{def:ser:1}
$f_e^{(i)}: \msg \rightarrow \Sigma_e^{(i)}$, 
\item \label{def:ser:2}
$\forall \; m_1, m_2 \in \msg$, if $f_e(m_1) = f_e(m_2)$, then $\forall i, \; f_e^{(i)}(m_1) = f_e^{(i)}(m_2)$,
\item \label{def:ser:3}
$\forall \; m_1, m_2 \in \msg$, if $f_e(m_1) \ne f_e(m_2)$, then $\exists i, \; f_e^{(i)}(m_1) \ne f_e^{(i)}(m_2)$, and 
\item \label{def:ser:4}
$\forall \; m \in \msg$, $e \in E$, $j \in \{1\ldots k\}$ there is some function $h_e^{(j)}$ such that 
$$
\cf_e^{(j)}(m) = h_e^{(j)}\left(\prod_{\hat{e} \in \In(e)} f_{\hat{e}}^{(1..j-1)}\right).
\footnote{Throughout this paper, when the operator $\prod$ is applied to functions rather then sets we mean it to denote the operation of forming an ordered tuple from an indexed list of elements.} 
$$
\end{enumerate}
We call such a $\Sigma_e^{(1...k)}, f_e^{(1...k)}$ a \emph{serialization} of $\Phi$. 
\end{definition}

The function $f_e^{(i)}$ describes the information sent on edge $e$ at time step $i$.  Item 2 requires that together the functions $f_e^{(1..k)}$ send no more information than $f_e$ and Item 3 requires that $f_e^{(1..k)}$ sends at least as much information as $f_e$.  Item 4 requires that we can compute $f_e^{(j)}$ given the information sent on all of $e$'s incoming edges at previous time steps.

In working with network codes, we will occasionally want to 
compare two network codes $\nwc, \, \nwc'$ such that $\nwc'$
``transmits all the information that is transmitted by $\nwc$.''  
In this case, we say that $\nwc'$ is an extension of $\nwc$,
and $\nwc$ is a restriction of $\nwc'$.

\begin{definition}
Suppose that $\Phi=(G,\msg,\{\Sigma_e\},\{\cf_e\})$ and
$\Phi'=(G,\msg,\{\Sigma'_e\},\{\cf'_e\})$ are two network
codes with the same sourced graph $G$ and the same message set
$\msg$.  We say that $\Phi$ is a restriction of $\Phi'$,
and $\Phi'$ is an extension of $\Phi$, if it is the case 
that for every $m \in \msg$ and $e \in E$, the value of 
$\cf'_e(m)$ completely determines the value of $\cf_e(m)$;
in other words, $\cf_e = g_e \circ \cf'_e$ for some 
function $g_e : \Sigma'_e \rightarrow \Sigma_e.$
\end{definition}
%
%
%<<<<<<< defs.tex
%When a code $\Phi$ is not serializable, it might be that some subset of the code is serializable.  
%This motivates the following
%definition.
%
%\begin{defn}
%For a network codes $\nwc = \networkcode{}$ and $\Phi' = \networkcode{'}$ we say that $\Phi'$ is a \emph{restriction} of $\Phi$ and $\Phi$ is an extention of $\Phi'$ if for all $m \in \msg$ $f'_e(m)$ is computable from $f_e(m)$.  $\Phi'$ is a \emph{serializable restriction} of $\Phi$ if $\Phi'$ is serializable, and likewise, $\Phi$ is a \emph{serializable restriction} of $\Phi'$ if $\Phi$ is serializable.    
%\end{defn}
%
%\Anote{This needs some lead in sentance but I don't know what to say}
The entropic vector of a network code gives a non-negative value for each subset of a network code.  The value of an edge set $F$ is the Shannon entropy of the joint distribution of the random variables associated with each element of $F$, as is formalized in the following definition.

\begin{defn}
Given a network code $\nwc = \networkcode{}, G = (V, E, S)$,
% and an edge set $F = \{e_1, e_2, \ldots, e_j\}\subset E \cup S$, 
the \emph{entropic vector} of $\nwc$ has coordinates $H(F)$ defined for each edge set $F = \{e_1,\ldots,e_j\} \subseteq E \cup S$ by:
$$ H(F) = H(e_1e_2\ldots e_j) =  \sum_{x_1 \in \Sigma_{e_1}, x_2 \in \Sigma_{e_2}, \ldots, x_j \in \Sigma_{e_j}} -p(x_1,x_2,\ldots, x_j)\log(p(x_1,x_2,\ldots, x_j)),$$ where the probabilities are computed assuming a uniform distribution over $\msg$.
\end{defn}

\section{Serializability and Entropy Inequalities}
\label{sec:2cycle}
Constraints imposed on the entropic vector alone suffice to characterize serializability for DAGs, but, the addition of one cycle causes the entopic vector to be an insufficient characterization.  We show that the entropic vector is not enough to determine serializability even on the 2-cycle by giving a serializable and non-serializable code with the same entropic vector.  
   
   \begin{figure}[h!]
   \centering
     \subfigure[Serializable]{
          \label{fig:same_Ha}
          \includegraphics[width=.4\textwidth]{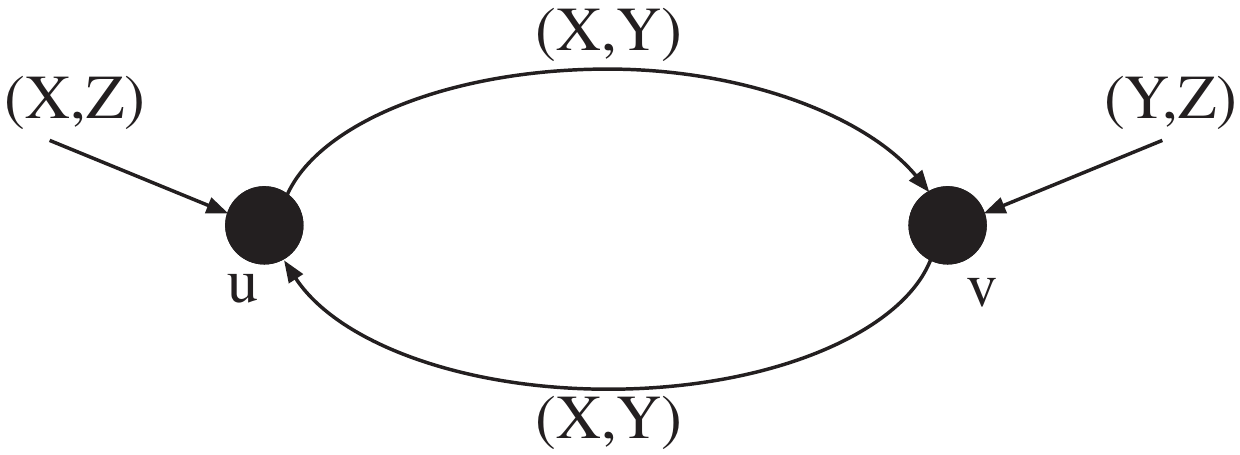}}
     \subfigure[Not Serializable]{
          \label{fig:same_Hb}
          \includegraphics[width=.4\textwidth]{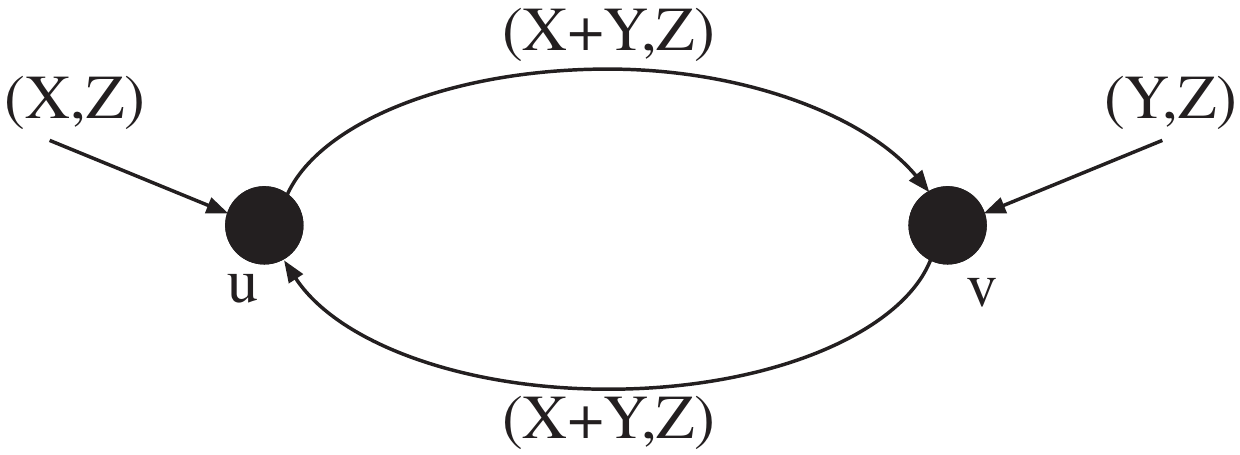}}\\
     \caption{Two network codes with the same entropy function.}
     \label{fig:same_H}
    \end{figure}

The two codes illustrated in Figure \ref{fig:same_H} apply to the message tuple $(X,Y,Z)$, where $X,Y,Z$ are uniformly distributed random variable over $\field_2$.  It is easy to check that the entropy of every subset of corresponding source and edge functions is the same, and thus the codes have the same entropic vector.  The code in Figure \ref{fig:same_Ha} is clearly serializable: at time step one we send $X$ on $(u,v)$ and $Y$ on $(v,u)$; then $Y$ on $(u,v)$ and $X$ on $(v,u)$.  On the other hand, the code in Figure \ref{fig:same_Hb} is not serializable because, informally, to send $X+Y$ on the top edge  requires that we already sent $X+Y$ on the bottom edge, and vice versa.  A formal proof that the code in Figure \ref{fig:offsetb} is not serializable can be obtained by applying the characterization of serializability in Theorem~\ref{thm:tfae-general}.

In order to use entropy inequalities to give tight upper bounds on network coding rates, we need
% don't need a set of inequalities that are necessary and sufficient, but rather 
an enumeration of the complete set of entropy inequalities implied by serializability, i.e.~a list of necessary and sufficient conditions for a vector $V$ to be the entropic vector of a serializable code.  (Note that it need not be the case that every code whose entropic vector is $V$ must be serializable.)  For the 2-cycle we can enumerate the complete set of inequalities.  In particular, we give four inequalities that must hold for any serializable code on the 2-cycle: two are a result of \emph{downstreamness} which is a condition that must hold for all graphs (it says that the entropy of the incoming edges of a vertex must be at least as much as the entropy of the incoming and outgoing edges together), the third is the \emph{Chicken and Egg} inequality due to \cite{CE}, and the fourth is a new inequality that we call the \emph{greedy} inequality.  It is equivalent to being able to complete the first iteration of our greedy algorithm in Section \ref{sec:serialize}.  We show that these four inequalities together with Shannon's inequalities are the only inequalities implied by serializability, in the following sense:

\begin{thm}
Given a rational-valued entropic vector, $V$, of a 2-cycle on nodes $u,v$, with source $x$ into node $u$, source $y$ into node $v$, and edges $a = (u,v)$ and $b = (v,u)$, there exists a serializable code that realizes $cV$, for some constant $c$, if and only if $V$ satisfies Shannon's inequalities, downstreamness ($H(abx) = H(bx)$, $H(aby) = H(ay)$), the Chicken and Egg inequality ($H(ab) \geq H(abx) - H(x) + H(aby) - H(y)$), and the greedy inequality ($H(a) + H(b) > H(ax) - H(x) + H(by) - H(y)$ when $H(a) + H(b) \ne 0$).
\label{thm:2cycle}
\end{thm}

Multiplication of the vector by a constant $c$ is a natural relaxation because the theorem becomes oblivious to the base of the logarithm we use to compute the Shannon entropy.

The proof of Theorem \ref{thm:2cycle} involves considering four cases corresponding to the relationship between $H(a)$ and $H(ax), H(ay), H(x), H(y)$ and between $H(b)$ and $H(bx), H(by) H(x), H(y)$.  Each case requires a distinctly different coding function to realize the entropic vector.  All the coding functions are relatively simple, involving only sending uncoded bits, and the XOR of two bits.  Most of the work is limiting the values of coordinates of the entropic vector based on the inequalities that the entropic vector satisfies.  The proof of Theorem \ref{thm:2cycle} is provided in Appendix \ref{ap:2cycle}.

The big open question left from this work is whether we can find a complete set of constraints on the entropic vector implied by the serializability of codes on arbitrary graphs.  We currently do not know of any procedure for producing such a list of inequalities.  
%\footnote{The greedy algorithm we present in Section \ref{sec:serialize} implicitly implies that a large set of inequalities must hold - similar in nature to our greedy inequality, but one for each iteration.  These inequalities are not constraints on the entropic vector because they involve comparing the entropy of common information which cannot be determined from entropy alone, but one needs to know the actual random variables.}  
Even if we had a conjecture for such a list, showing that it is complete is likely to be quite hard.  If we have more than three sources, just determining the possible dependencies between sources is difficult because they are subject to non-Shannon information inequalities.

 \section{A Characterization of Serializability}
\label{sec:serialize}

\subsection{Linear Codes}
\label{sec:linear}

A characterization of serializability for linear network codes is simpler than the general
case because it relies on more standard algebraic 
tools.  Accordingly, we treat this case first before moving
on to the general case.  Throughout this section
we use $V^*$ to denote the dual of a vector space
over a field $\field$, and $f^*$ to denote the adjoint
of a linear transformation $f$.  For an edge $e$
with alphabet $\Ealph$ and coding function $f_e :
\msg \rightarrow \Ealph$, we use $T_e$ to denote
the linear subspace $f_e^*(\Ealph^*) \subseteq \msg^*$.

Though it is impossible to characterize the serializability of a network code in terms of its entropic vector, computationally there is a straightforward solution.  In polynomial time we can either determine a serialization for a code or show that no serialization exists using the obvious algorithm: try to serialize the code by ``sending new information when possible.''  When we can no longer send any new information along any edge we terminate.  If we have sent all the information required along each edge, then the greedy algorithm finds a serialization; otherwise, we show that no serialization exists by presenting a succinct certificate of non-serializability.  Though our algorithm is straightforward, we believe that the change in mindset from characterizing codes in terms of the entropic vector is an important  one, and that our certificate of non-serializability (see Definition~\ref{def:iv}) furnishes an effective tool for addressing other questions about serializability, as we shall see in later sections.  
%And, the ease of characterizing serializability in this way is surprising given how difficult it is to find constriants on the entropic vector implied by serializability.  

Given a network code $\nwc = \networkcode{}$, with coding functions over the field $\field$, our greedy algorithm (pseudocode, \CallSer, is given
in Appendix~\ref{ap:linear}), constructs a set of edge functions $f^{(1..k)}_e$ and alphabets $\Sigma_e^{(1..k)}$ for each edge.  
These objects are constructed iteratively, defining the edge alphabets $\Sigma_e^{(i)}$ and coding functions $f^{(i)}_e$ in the $i^{th}$ iteration. 
Throughout this process,
we maintain a pair of linear subspaces $A_e,\, B_e \subseteq \msg^*$ 
for each edge $e=(u,v)$ of $G$.  $A_e$ is the
linear span\footnote{If $\{V_i : i \in \mathcal{I}\}$ is a collection of linear subspaces of a vector space $V$, their linear span is the minimal linear subspace containing the union $\bigcup_{i \in \mathcal{I}} V_i.$  We denote the linear span by $\spn_{i \in \mathcal{I}} V_i.$} of all the messages transmitted on $e$ so far,
and $B_e$ is intersection of $T_e$ with 
the linear span of all the messages transmitted to
$u$ so far.  (In other words, $B_e$ spans all the messages that could currently be sent on $e$ without receiving
any additional messages at $u$.)  In the $i^{th}$ iteration, if there exists an edge $e'$ such that $B_{e'}$ contains a 
dual vector $x_{e'}$ that does not belong to $A_{e'}$, then we create coding function $f_e^{(i)}$ for all $e$.  The coding function of $f_{e'}^{(i)}$ is set to be $x_{e'}$ and its alphabet is set to be 
$\field$.  For all other edges we set $f_e^{(i)} = 0$.  This process
continues until $B_e=A_e$ for every $e$.  At that point,
we report that the code is serializable if and only if
$A_e=T_e$ for all $e$.  At the end, the algorithm returns
the functions $f_e^{(1..k)}$ and the alphabets $\Sigma_e^{(1..k)}$, where $k$ is the number of iterations of the algorithm, as well as the subspaces $\{A_e\}$.  If the code was not serializable, then $\{A_e\}$ is
interpreted as a certificate of non-serializability
(a ``non-trivial information vortex'')
as explained below.

\CallSer runs in time polynomial in the size of the coding functions of $\nwc$.  In every iteration of the while loop we increase the dimension of some $A_e$ by one.  $A_e$ is initialized with dimension zero and can have dimension at most $\dim(T_e)$.  Therefore, the algorithm goes through at most $\sum_{e \in E} \dim(T_e)$ iterations of the while loop.  Additionally, each iteration of the while loop, aside from constant time assignments, computes only intersections and spans of vector spaces, all of which can be done in polynomial time.

To prove the algorithm's correctness, we define the following certificate of non-serializability.

\begin{defn} \label{def:iv}
An \emph{information vortex (IV)} of a network code consists
of a linear subspace $W_e \subseteq \msg^*$ for each edge
$e$, such that:
\begin{enumerate}
\item  For a source edge $s$,
$W_s = T_s.$
\item  For every other edge $e$,
$
W_e = T_e \intsct \left( \Spn_{e' \in \In(e)} W_{e'} \right).
$
\end{enumerate}
An information vortex is \emph{nontrivial} if 
$W_e \ne T_e$ for some edge $e$.
\end{defn}

We think of $W_e$ as the information that we can send over $e$ given that its incoming edges, $e' \in In(e)$, can send $W_{e'}$.  In our analysis of the greedy algorithm, we show that the messages the greedy algorithm succeeds in sending (i.e.,~the linear subspaces $\{A_e\}$) form an IV and it is non-trivial if and only if the code isn't serializable.

The following theorem shows the relationship between IVs, serialization, and the greedy algorithm.  The proof can be found in Appendix~\ref{ap:linear}.

\begin{thm}
For a network code $\nwc = \networkcode{}$, the following are equivalent:
\begin{enumerate}
\item $\nwc$ is not serializable
\item \CallSer returns $\{A_e\}$ \st $\exists \, e, A_e \ne T_e$
\item $\nwc$ has a non-trivial information vortex
\end{enumerate}
\label{thm:TFAE}
\end{thm}

In Section \ref{sec:deficit} and Section \ref{sec:asymptotic} we will see that information vortices provide a concise way for proving the non-serializability of a network code.  Moreover, the notion of an information vortex was critical to our discovery of the result in Section \ref{sec:asymptotic}. 

\subsection{General codes}
\label{sec:general}

\renewcommand{\intsct}{\cap}

Our characterization theorem extends to the case of general 
network codes, provided that we generalize the greedy algorithm
and the definition of information vortex appropriately.  
The message space $\msg$ is no longer a vector space, so instead 
of defining information vortices using the vector space $\msg^*$ 
of all linear functions on $\msg$, we use the Boolean algebra
$2^{\msg}$ of all binary-valued functions on $\msg$.
We begin by recalling some notions from the theory of Boolean 
algebras.
\begin{definition} \label{def:balg}
Let $S$ be a set.  The Boolean algebra 
$2^S$ is the algebra consisting of all
$\{0,1\}$-valued functions on $S$, under
{\sc and} ($\wedge$), {\sc or} ($\vee$),
and {\sc not} ($\neg$).
If $f : S \rightarrow T$ is a function, then
the \emph{Boolean algebra generated by $f$},
denoted by $\balg{f}$, 
is the subalgebra of $2^S$ consisting of all
functions $b \circ f$, where $b$ is a 
$\{0,1\}$-valued function on $T$.
If $A_1,A_2$ are subalgebras of a Boolean
algebra $\bal,$ their intersection $A_1 \intsct A_2$
is a subalgebra as well.  Their union is not,
but it generates a subalgebra that we will
denote by $A_1 \spn A_2.$

If $S$ is a finite set and $\bal \subseteq 2^{S}$ is
a Boolean subalgebra, then there is an equivalence 
relation on $S$ defined by setting $x \sim y$ if 
and only if $b(x)=b(y)$ for all $b \in \bal$.
The equivalence classes of this relation are called
the \emph{atoms} of $\bal$, and we denote the set
of atoms by $\atoms(\bal)$.
There is a canonical function $f_{\bal} : S
\rightarrow \atoms(\bal)$
that maps each element to its equivalence class.
Note that $\bal = \balg{f_{\bal}}$.
\end{definition}
The relevance of Boolean subalgebras
to network coding is as follows.  
A subalgebra $\bal \subseteq 2^{\msg}$
is a set of binary-valued 
functions, and can be interpreted as describing
the complete state of knowledge of a party
that knows the value of each of these functions
but no others.  In particular, if
a sender  knows the value of $f(m)$
for some function $f:\msg \rightarrow T$,
then the binary-valued messages this sender
can transmit given its current state of knowledge
correspond precisely to the elements
of $\balg{f}$.  This observation supplies the raw
materials for our definition of the greedy algorithm
for general network codes, which we denote by \GenSer.

As before, the edge alphabets and coding functions
are constructed iteratively, with $\Sigma_e^{(i)}$ and $f^{(i)}_e$ 
defined in the $i^{th}$ iteration of the main loop.
Throughout this process,
we maintain a pair of Boolean subalgebras $A_e,\, B_e \subseteq 2^{\msg}$ 
for each edge $e=(u,v)$ of $G$.  $A_e$ is generated by
all the messages transmitted on $e$ so far,
and $B_e$ is intersection of $\balg{f_e}$ with 
the subalgebra generated by all messages transmitted to
$u$ so far.  (In other words, $B_e$ spans all the binary-valued 
messages that could currently be sent on $e$ without receiving
any additional messages at $u$.)  In the $i^{th}$ iteration, 
if there exists an edge $e'$ such that $B_{e'}$ contains a 
binary function $x_{e'} \not\in A_{e'}$, 
then we create a binary-valued 
coding function $f_e^{(i)}$ for all $e$,
which is set to be $x_{e'}$ if $e=e'$ and the constant
function $0$ if $e \neq e'$.  
This process
continues until $B_e=A_e$ for every $e$.  At that point,
we report that the code is serializable if and only if
$A_e=\balg{f_e}$ for all $e$.  At the end, the algorithm returns
the functions $f_e^{(1..k)}$ and the alphabets $\Sigma_e^{(1..k)}$, 
where $k$ is the number of iterations of the algorithm, as well 
as the subspaces $\{A_e\}$.  The pseudocode for this algorithm
\GenSer\  is presented in  Appendix~\ref{ap:general}.

% as we build up the DAG $G'$ and its coding functions,
% we keep track of a Boolean subalgebra $\bal_v \subseteq 2^{\msg}$ 
% for each vertex $v$, representing the information already known
% to $v$, as well as a Boolean subalgebra $\bal_e \subseteq 2^{\msg}$ 
% for each edge $e$ of $G$, representing the functions
% in $\balg{f_e}$ whose value is determined
% by the information that has already been sent on $e$
% (i.e., the coding functions of the existing DAG edges 
% that map to $e$).
% As long as some $u$ knows the value of a 
% binary-valued function
% whose value is not determined by the information already
% sent on $e=(u,v)$ (i.e. an element of $\bal_u$ but not $\bal_e$), 
% if it is allowed to transmit that function on $e$ (i.e. it
% belongs to $\balg{f_e}$) then it does so.  In other
% words, the algorithm terminates only when $\bal_e = 
% \bal_u \intsct \balg{f_e}$ for all $e \in E$.  
% At termination it outputs ``serializable'' if 
% $\bal_e = \balg{f_e}$ for all $e \in E$, otherwise
% ``not serializable.''
% The pseudocode for the algorithm itself is extremely
% similar to the pseudocode for Algorithm~\ref{alg:greedy-linear},
% and is presented in full in Appendix~\ref{ap:general}.  

If $\Phi$ has finite alphabets, then \GenSer\ must 
terminate because the
total number of atoms in all the Boolean algebras
$\bal_e \; (e \in E)$ is strictly increasing
in each iteration of the main loop, so 
$\sum_{e \in E} |\Sigma_e|$ is an upper bound on
the total number of loop iterations.  In implementing
the algorithm, each 
of the Boolean algebras can be represented as
a partition of $\msg$ into atoms, and all of the
operations the algorithm performs on Boolean
algebras can be implemented in polynomial time
in this representation.  Thus, the running time
of \GenSer\ is polynomial in $\sum_{e \in E} |\Sigma_e|.$
In light of the algorithm's termination 
condition, the following definition is natural.
\begin{definition} \label{def:gdiv}
If $G=(V,E,S)$ is a sourced graph, 
a \emph{generalized information vortex (GIV)}
in a network code $\Phi = (G,\msg,\Sigmas,\fs)$
is an assignment of Boolean subalgebras $\bal_e \subseteq 2^{\msg}$
to every $e \in E \cup S$, satisfying:
\begin{enumerate}
\item  \label{gdiv-s} $\bal_s = \balg{f_s}$ for all $s \in S$;
\item  \label{gdiv-e} $\bal_e = 
\left( \spn_{\hat{e} \in \In(u)} \bal_{\hat{e}} \right)
\intsct \balg{f_e}$ 
for all $e = (u,v) \in E$.
% \item  \label{gdiv-v} $\bal_v = \spn_{e \in \In(v)} \bal_e$ for all $v \in V$.
\end{enumerate}
A GIV is \emph{nontrivial} if
$\bal_e \neq \balg{f_e}$ for some $e \in E$.
A tuple of Boolean subalgebras
$\gdiv = (\bal_e)_{e \in E \cup S}$ is a 
\emph{semi-vortex} if it satisfies (\ref{gdiv-s})
but only satisfies one-sided containment in (\ref{gdiv-e}), i.e.,
\begin{enumerate} \setcounter{enumi}{2}
\item  \label{semi-vortex}
$\bal_e \subseteq
\left( \spn_{\hat{e} \in \In(u)} \bal_{\hat{e}} \right)
\intsct \balg{f_e}$ for
all $e = (u,v) \in E$.  
\end{enumerate}
If $\gdiv = (\bal_e)$ and $\ggdiv = (\bal'_e)$ 
are semi-vortices, we say that $\gdiv$
is contained in $\ggdiv$ if $\bal_e \subseteq \bal'_e$
for all $e$.
\end{definition}
In Appendix~\ref{ap:general} we prove a series of statements
(Lemmas~\ref{lem:restriction}-\ref{lem:termination})
showing that:
\begin{itemize}
\item Semi-vortices are in one-to-one  correspondence with
restrictions of $\nwc$.  The correspondence maps a 
semi-vortex $(\bal_e)_{e \in E \cup S}$ to the network
code with edge alphabets $\atoms(\bal_e)$
and coding functions given by the canonical maps
$\msg \rightarrow \atoms(\bal_e)$ defined in
Definition~\ref{def:balg}.
\item There is a set of semi-vortices
corresponding to serializable restrictions of 
$\nwc$ under this correspondence.  
They can be thought of as representing
\emph{partial serializations} of $\nwc$.
\item  There is a set
of semi-vortices corresponding to GIV's of $\nwc.$
These can be thought of as \emph{certificates
of infeasibility} for serializing $\nwc$.
% \item
% There is a unique semi-vortex $\gdiv$ in the intersection
% of these two sets: it is a GIV which is simultaneously
% a partial serialization and a certificate that no further
% serialization is possible.  It is a maximal
% partial serialization and a minimal GIV, when
% both of those sets are ordered by containment.
% $\Phi$ is serializable if
% and only if this $\gdiv$ is trivial.
\item \GenSer\ computes a semi-vortex $\gdiv$ 
which is both a GIV and a partial serialization.% of $\nwc$.
\end{itemize}
These lemmas combine to yield a 
``min-max theorem'' showing that
the every network code has a maximal serializable
restriction that coincides with its minimal GIV,
as well as an analogue of Theorem~\ref{thm:TFAE};
proofs of both theorems are in Appendix~\ref{ap:general}.
\begin{thm} \label{thm:minmax}
In the ordering of semi-vortices by containment, the
ones corresponding to partial serializations
have a maximal element and the GIV's have a minimal
element.  These maximal and minimal elements coincide, 
and they are both equal to the
semi-vortex $\gdiv=(A_e)_{e \in E \cup S}$ computed by \GenSer.
\end{thm}
%Finally, we have the following analogue of Theorem~\ref{thm:TFAE}.
\begin{thm} \label{thm:tfae-general}
For a network code $\Phi$ with finite alphabets, 
the following are equivalent.
\begin{enumerate}
\item \label{tfae-general-ser} $\Phi$ is serializable.
\item \label{tfae-general-greedy}
\GenSer\ outputs $\{\bal_e\}_{e \in E}$ \st $\forall \, e, A_e = \balg{f_e}.$
\item \label{tfae-general-gdiv} $\Phi$ has no nontrivial
GIV.
\end{enumerate}
\end{thm}

\section{The Serializability Deficit of Linear Network Codes}
\label{sec:deficit}

The min-max relationship between serializable restrictions and 
information vortices (Theorem~\ref{thm:minmax}) is reminiscent
of classical results like the max-flow min-cut theorem.  However,
there is an important difference: one can use the minimum cut in
a network to detect \emph{how far} a network flow problem is from
feasibility, i.e.~the minimum amount by which edge capacities
would need to increase in order to make the problem feasible.
In this section, we will see that determining how far a network
code is from serializability is more subtle: two network codes
can be quite similar-looking, with similar-looking minimal 
information vortices, yet one of them can be serialized by
sending only one extra bit while the other requires many more
bits to be sent.

We begin with an example to illustrate this point.  The codes in Figure \ref{fig:offset} apply to the message tuple $(X_1, \ldots, X_n, Y_1, \ldots, Y_n)$ where $X_i, Y_i$ are independent, uniformly distributed random variables over $\field_2$.  The codes in Figures \ref{fig:offseta} and \ref{fig:offsetb} are almost identical; the only difference is that the code in Figure \ref{fig:offseta} has one extra bit along the top edge.  The code in Figure \ref{fig:offseta} is serializable: transmit $X_1$ along $(u,v)$, then $X_1+Y_1$ on edge $(v,u)$, then $X_2+Y_1$ on $(u,v)$, ... ,$X_n+Y_n$ on $(v,u)$, and finally $X_1+Y_n$ on $(u,v)$.  On the other hand, the code in Figure \ref{fig:offsetb} is not serializable, which can be seen by applying our greedy algorithm.

% In Section~\ref{sec:char-linear} we saw 
\begin{figure}[h!]
     \centering
     \subfigure[Serializable]{
          \label{fig:offseta}
          \includegraphics[width=.4\textwidth]{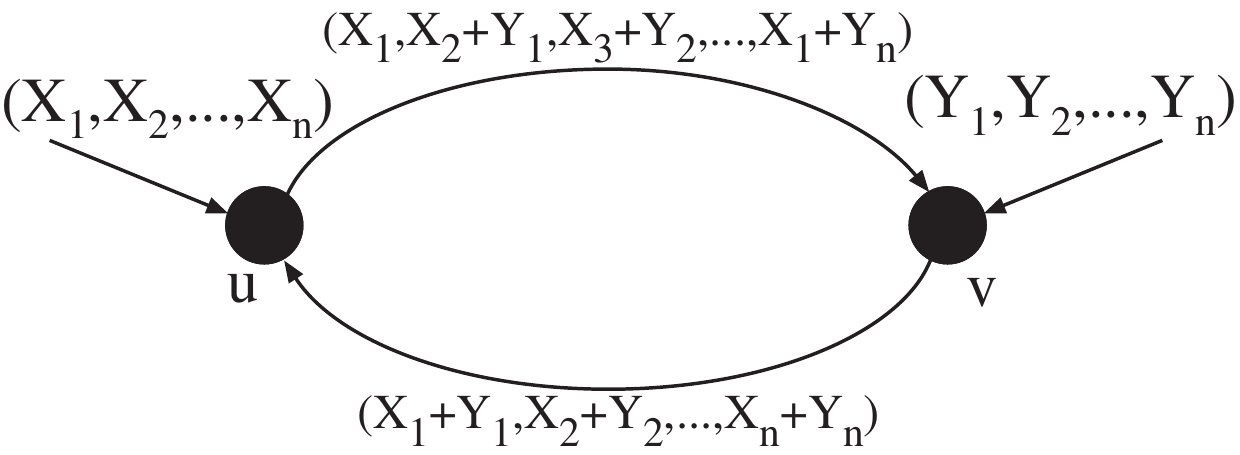}}
     \hspace{.2in}
     \subfigure[Not Serializable]{
          \label{fig:offsetb}
          \includegraphics[width=.4\textwidth]{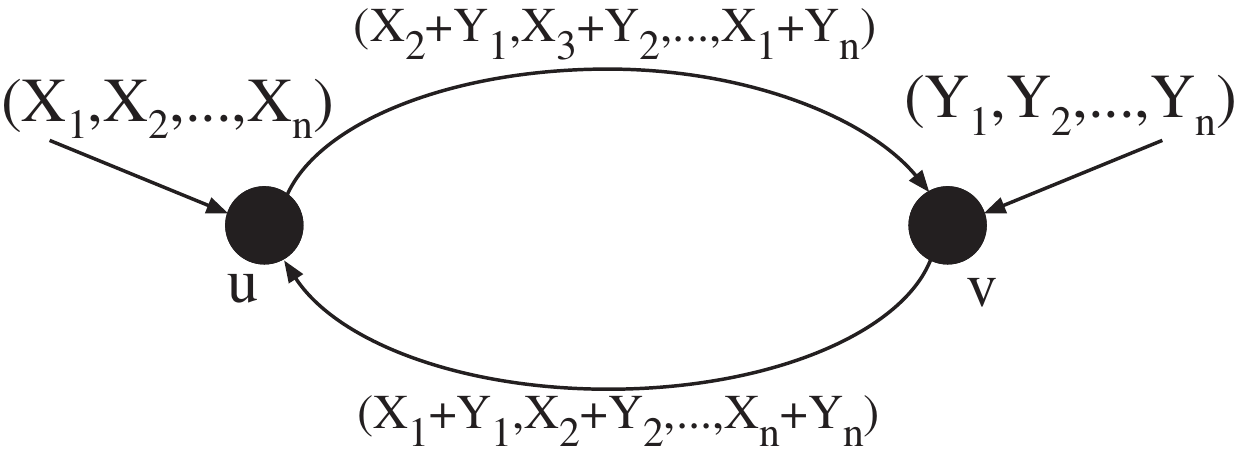}}\\
     \vspace{.01in}
     \caption{Two almost identical network codes.}
     \label{fig:offset}
\end{figure}

% Our greedy algorithm answers just the yes/no version of the serializability question.  But, given this example, even better would be an algorithm that tells us how ``close'' a code is to being serializable.  T
Thus, the code in Figure \ref{fig:offsetb} is very close to serializable because we can consider an extension of the code in which we add one bit\footnote{In this section, for simplicity, we refer to one scalar-valued linear function on an $\field$-vector space as a ``bit'' even if $|\field|>2$.} to the edge $(u,v)$ to obtain the code in Figure \ref{fig:offseta} that is serializable.  On the other hand, there are similar codes that are very far from being serializable.  If we consider the code with the same sources and $f_{(u,v)} = f_{(v,u)} = \prod_{i=1}^n X_i+Y_i$, its edge alphabets have the same size and its minimal information vortex is identical, yet any serializable extension requires adding $n$ bits.  To completely characterize serializability we would like to be able to separate codes that are close to serializable from those that are far.  This motivates the following definition.

\begin{definition}
For a network code $\Phi = \networkcode{}$ and an extension $\Phi' = \networkcode{'}$, the \emph{gap} of $\Phi'$, defined by $\hgap{\Phi'} = \sum_{e \in E} \log_2 |\Sigma'_e| - \log_2|\Sigma_e|$,
represents the combined number of extra bits transmitted
on all edges in $\Phi$ as compared to $\Phi$. The \emph{serializability
deficit} of $\Phi$, denoted by $\gsdef(\Phi)$, 
is defined to be the minimum of 
$\hgap{\Phi'}$ over all serializable extensions $\Phi'$
of $\Phi$.  The \emph{linear serializability deficit} of a linear code $\Phi$, denoted $\sdef(\Phi)$, is the minimum of $\hgap{\Phi'}$ over all linear serializable extensions $\Phi'$.
\end{definition}

Unfortunately, determining the serialization deficit is much more difficult than simply determining serializability.

\begin{thm}
Given a linear network code $\nwc$, it is NP-hard to approximate the size of the minimal linear serializable extension of $\nwc$.
Moreover, there is a linear network code $\Phi$ and a positive integer $n$ such that $\sdef(\nwc^n)/(n \sdef(\nwc)) < \bigo(\frac{1}{\log_2(n)})$.  
\label{thm:gap}
\end{thm}
Both statements in the theorem follow directly from the following
lemma.

\begin{lem}
Given a hitting set instance $(N,S)$ with universe $N$, $|N| = n$, and subsets $S \subseteq 2^N$, an optimal integral solution $k$, and an optimal fractional solution $\frac{z_1}{q},\frac{z_2}{q},..,\frac{z_n}{q}$, with $\sum_{i=1}^n \frac{z_i}{p_1} = \frac{p}{q}$, in polynomial time we can construct a linear network code such that
% \begin{enumerate}
% \item 
$\sdef(\nwc) = k$, but
% \item 
$\sdef(\nwc^q) \leq p$.
% \end{enumerate}
\label{lem:reduc}
\end{lem}

{\em {Proof sketch.}}
\begin{floatingfigure}[1]{0.5\textwidth}
   	  \centering
    	 \includegraphics[width=0.45\textwidth]{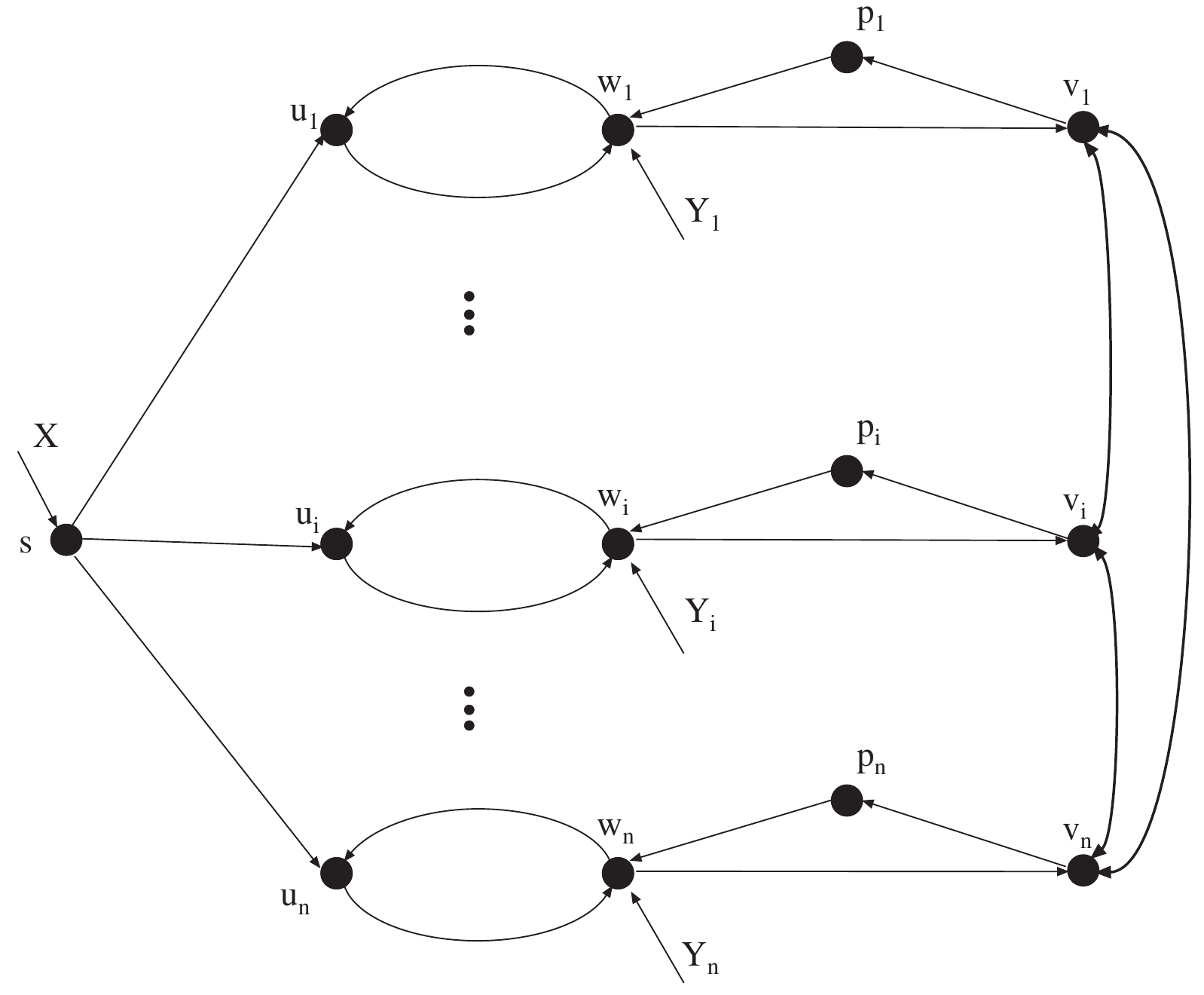}
         \caption{The reduction from Hitting Set}
%     	\caption{An illustration of the vertices, edges, and codes in the network coding instance that are created for one element $i\in N$ of the universe of our hitting set instance.  We have a similar diagram for each $i \in N$.  The interaction between the elements $i \in N$ happens only on the vertices $V$.}
     \label{fig:reduc}
\end{floatingfigure}
The full proof of the Lemma is in Appendix~\ref{ap:linear}.
Here, due to space limitations, we merely sketch the main ideas.
The graph used in the reduction is illustrated pictorially 
in Figure~\ref{fig:reduc}.  Given a hitting set instance
$(N,S)$ we create a network code with one source for each
$i \in N$ (source message denoted by $\vec{Y}_i$) and a super-source
$s$ (source message denoted $(\vec{X}_1,\vec{X}_2,\ldots,\vec{X}_n)$).  
The symbols $\vec{X}_i,\vec{Y}_i$ don't refer to % scalars in $\field$
% (henceforth denoted ``bits'' for ease of terminology, 
% even though $\field$ may not be $\field_2$)
bits, 
but actually to blocks of $n_i$ bits, where $n_i$ is the
number of sets in $S$ containing $i$;  each of the
bits in $\vec{X}_i$ or $\vec{Y}_i$ corresponds to one of the sets that 
$i$ belongs to.  For
each $i \in N$ we use a gadget consisting of a
$2$-cycle on vertices $u_i,w_i$, with $\vec{Y}_i$ feeding into 
$w_i$ and $\vec{X}_i$ feeding from the super-source $s$ into $u_i$.
The edges between $u_i$ and $w_i$ are a copy of the
gadget in Figure~\ref{fig:offset}.  We exploit the fact
that sending one extra bit in this gadget allows the
information vortex in the gadget to ``unravel'', leading
to transmission of all the bits encoded on the edges of
the $2$-cycle.
The $2$-cycle $(u_i,w_i)$ participates in a larger 
$4$-vertex gadget $\{u_i,w_i,p_i,v_i\}$ corresponding
to the element $i$.  The role of $v_i$ is to 
participate in ``set gadgets'', where the gadget
corresponding to a set $A$ consists of a
bidirected clique on all the vertices $\{v_j | j \in A\}.$
The role of $p_i$ is less important; it plays a 
necessary part in disseminating bits to leftover 
parts of the network after the ``important'' parts have
been serialized.  If
there is a hitting set of size $k$ then we send
one bit on each of the $2$-cycles $(u_i,w_i)$
corresponding to elements $i$ in the hitting set.
This ``unlocks'' the bits that were locked up
in those $2$-cycles, which allows a sufficient 
amount of information to flow into the set gadgets
that they become serialized.  The vertices $p_j$
are then used for disseminating the remaining
bits to the unused $2$-cycles $(u_j,w_j)$ where
$j$ did not belong to the hitting set.  

To prove, conversely, that a serializability deficit 
of at most $k$ implies that there is a hitting set of size
$k$, we make use of the fact that the network code
constructed by our reduction has a large number of 
information vortices, one for each pair 
consisting of an element of $N$ and a set $S$ 
that it belongs to.  If $C$ is the set of all $i$
such that an extra bit is transmitted somewhere in
the 4-vertex gadget for $i$, and $C$ fails to
contain an element of some set $A$, then this 
in turn implies that one of the aforementioned 
information vortices remains an
information vortex in the extension of the code.
Thus, $C$ must be a hitting set.

The more difficult step in proving Lemma~\ref{lem:reduc}
lies in showing that fractional solutions of the hitting
set problem can be transformed into efficient
serializable extensions of $\Phi^q$.  For this,
we make use of the fact that the edge alphabets
in $\Phi^q$ can be regarded as $\Sigma_e \otimes \field^q.$
and their duals can be regarded as $\Sigma_e^*
\otimes (\field^*)^q.$  If $|\field|$ is large
enough, then the uniform matroid $U_{q,p}$ is
representable as a set $\{t_1,\ldots,t_p\}$ of $p$ vectors in
$(\field^*)^q$.  For each of the ``fractional elements''
$z_i/q$ in the fractional hitting set, we 
send $z_i$ bits  of the form $x \otimes t$ in the extension
of $\Phi^q,$ where $t$ is one of the elements
of our matroid representation of $U_{q,p}$ 
in $(\field^*)^q$
and $x$ is the bit that we would have sent
for element $i$ in the hitting set reduction
described two paragraphs earlier.  The fact
that $\sum_i z_i=p$ implies that every
element of the matroid representation is used
exactly once in this construction.  The
fact that we have a fractional set cover
implies that in each set gadget, we receive
extra bits corresponding to $q$ distinct
elements of $S$.  Since these elements
are a basis for $(\field^*)^q$, it is then
possible to show that they combine to allow
a serialization of all the ``missing bits'' 
in that gadget, and from there we finish
serializing the entire network code $\Phi^q$
as before.

% A critical ingredient of the reduction is a gadget similar to the example in Figure \ref{fig:offset}.  We exploit the fact that we can transform a code where no bits can be sent to a code where $n$ bits can be sent with the addition of only one bit.  % This helps us limit the space of useful tibs - tibs that ``release'' all the information in the 2-cycle will be much more valuable than any other. 

% \input{fractionalHS}

\section{Asymptotic Serializability}
\label{sec:asymptotic}

The results in the previous section showed both that there are non-serializable codes with large edge alphabets that become serializable by adding only one bit (example in Figure \ref{fig:offset}) and that the serialization deficit can behave sub-additively when we take the $n$-fold cartesian product of $\Phi$ (Theorem \ref{thm:gap}).  This prompts the investigation of whether there exists a code that isn't serializable, but the $n$-fold parallel repetition of the code can be serialized by extending it by only a constant number of bits, and thus it is essentially indistinguishable from serializable.  We formalize this idea with the following definition. 

\begin{definition} \label{asympt}
A network code $\Phi$ is \emph{asymptotically serializable}
if $\lim_{n \rightarrow \infty} \frac1n \sdef(\Phi^n)/\sdef(\Phi) = 0$ where $\Phi^n$ is $n$-fold cartesian product of $\nwc$ with cartesian product define in the obvious way. 
\end{definition}

If one is using a network code to transmit
infinite streams of data by chopping each stream
up into a sequence of finite blocks and applying
the specified coding functions to each block, then
an asymptotically serializable network code is almost
as good as a serializable one, since it can be serialized
by adding a side channel of arbitrarily small bit-rate
to each edge of the network.

Despite indications to the contrary in Section \ref{sec:deficit}, we show that any non-serializable linear code is not asymptotically serializable via the following theorem. 

\begin{thm}
\label{thm:asymSer}
For a linear network code $\nwc = \networkcode{}$ over a field $\field$, then $\sdef(\nwc^n) \geq cn$ where $c$ is a constant dependent on $\nwc$. 
\end{thm}

The proof of the theorem considers the alphabets of the $n$-fold product of $\Phi$ as elements of a tensor product space.  Using this machinery, we show that information vortices in the graph are preserved if we don't increase the amount of information we send down some edge by order $n$ bits.  More specifically, if $\{W_e\}$ is a non-trivial information vortex in $\nwc$, and $e$ is an edge such that $\dim(W_e) < \dim(T_e) = m$, then if we add some edge function $f$ to every edge in the graph, the information vortex remains non-trivial as long as the dimension of $f$ is less than $mn$.  A complete proof is provided in Appendix \ref{ap:asymptotic}.

\bibliographystyle{plain}
\bibliography{bib}

\appendix
\section{Proof of Theorem \ref{thm:2cycle}}
\label{ap:2cycle}

\begin{thm}[Theorem \ref{thm:2cycle} restated]
Given a rational-valued entropic vector, $V$, of a 2-cycle on nodes $u,v$, with source $x$ into node $u$, source $y$ into node $v$, and edges $a = (u,v)$ and $b = (v,u)$, there exists a serializable code that realizes $cV$, for some constant $c$, if and only if $V$ satisfies Shannon's inequalities, downstreamness ($H(abx) = H(bx)$, $H(aby) = H(ay)$), the Chicken and Egg inequality ($H(ab) \geq H(abx) - H(x) + H(aby) - H(y)$), and the Greedy inequality ($H(a) + H(b) > H(ax) - H(x) + H(by) - H(y)$ when $H(a) + H(b) \ne 0$).
\end{thm}

Throughout this proof it will often be convenient to refer to the conditional entropy of two sets of edges.

\begin{defn}
For two subsets of edges $F = \{e_1,e_2,\ldots, e_j\}$ and $F'= \{e'_1,e'_2,\ldots, e'_k\}$, the conditional entropy of $F$ given $F'$, denoted $H(F|F') = H(e_1e_2 \ldots e_j|e'_1e'_2 \ldots e'_k) = H(FF') - H(F')$.
\end{defn} 

We first show that all the inequalities are necessary.  Shannon's inequalities must hold for the entropic vector of any set of random variables.  Downstreamness (term coined by \cite{ARL-thesis}) was shown to be necessary even for DAGS by Yeung and Zhang~\cite{YZ99}.  Harvey et al.~\cite{CE} showed that the Chicken and Egg inequality is necessary.  Thus, it remains to show that our greedy inequality is a necessary condition for serializability.   
\begin{lem}
The inequality $H(a) + H(b) > H(ax) - H(x) + H(by) - H(y) = H(a|x) + H(b|y)$ holds for any serializable code on a 2-cycle when $H(a)+H(b) >0$. 
\end{lem}

\begin{proof}
Suppose there is a serializable code such that $H(a) + H(b) \leq H(a|x) + H(b|y)$ and $H(a)+H(b) > 0$.  Because conditioning reduces entropy, $H(a) \ge H(a|x)$, and likewise $H(b) \ge H(b|y)$.  These three inequalities together imply that $H(a) = H(a|x)$ and $H(b) = H(b|y)$.  It follows from the definition of serializability and $H(a)+H(b) > 0$ that there exists a non-zero $f^{(i)}_a$ or $f^{(i)}_b$.  Let $i^*$ be the smallest such $i$ and let $f^{(i^*)}_a$ be the associated non-zero coding function (the choice of $a$ is without loss of generality).  We can rewrite $H(a|x)$ as $H(f_a^{(i^*)}|x) + H(a|f_a^{(i^*)}x)$.  $H(f_a^{(i^*)}|x) = 0$ because $i^*$ is the smallest such $i$ implies that $f_a^{(i^*)}$ is computed soley from $x$.   But, this gives us that $H(a) = H(a|f_a^{(i^*)}x)$, which is a contradiction to $f^{(i)}_a$ non-zero.
\end{proof}

To prove the other direction of Theorem \ref{thm:2cycle} we will use a case based analysis, but first we make a few observations to bound the cases we need to consider. 

\begin{obs}
The following ten values completely determine the entropic vector of the 2-cycle: $I(x;y), H(x|y), H(y|x), H(a|x), H(b|x), H(a|y), H(b|y), H(a), H(b), H(ab)$.
\label{10vals}
\end{obs}
\begin{proof}
Due to downstreamness and Shannon's inequalities the following equations hold: $H(xy) = I(x;y) + H(x|y) + H(y|x), H(axy) = H(bxy) = H(abxy) = H(xy), H(y) = I(x;y) + H(y|x), H(x) = I(x;y) + H(x|y), H(ax) = H(a|x) + H(x), H(aby) = H(ay) = H(a|y) + H(y), H(abx) = H(bx)= H(b|x) + H(x), H(by) = H(b|y) + H(y)$.  This implies that the value of all 15 non-zero elements of the entropic vector are determined by the 10.
\end{proof}

%\Anote{do I need to explain downstreamness?}

\begin{obs} 
$H(b|x) \ge H(a|x)$
\label{obs:bx_g_ax}
\end{obs}

\begin{proof}
$H(bx) = H(abx) \geq H(ax)$ by downstreamness and then monotonicity.
\end{proof}

\begin{obs}
$H(a|y) \geq H(b|y)$
\label{obs:ay_g_by}
\end{obs}

\begin{proof}
Parallel to proof of observation \ref{obs:bx_g_ax}
\end{proof} 

\begin{obs}
$\max(H(a|x),H(a|y)) \leq H(a) \leq H(a|x) + H(a|y) + I(x;y)$
\label{obs:boundHa}
\end{obs}

\begin{proof}
$H(a) \leq H(a|x) + H(a|y) + I(x;y)$: Apply submodularity on $ax$ and $ay$ to get $H(ax) + H(ay) \geq H(axy) + H(a) = H(xy) + H(a)$, then subtract $H(x) + H(y)$ from both sides.\\
$H(a) \geq \max(H(a|x),H(a|y)):$ $H(a) \geq H(a|x)$ and $H(a) \geq H(a|y)$ because conditioning reduces entropy.
\end{proof}

\begin{obs}
$\max(H(b|x),H(b|y)) \leq H(b) \leq H(b|x) + H(b|y) + I(x;y)$
\label{obs:boundHb}
\end{obs}

\begin{proof}
Parallel to proof of observation \ref{obs:boundHa}
\end{proof} 

\begin{obs}
$H(b|x)+H(a|y) \leq H(ab) \leq H(b|x)+H(a|y)+I(x;y)$
\label{obs:boundHab}
\end{obs}

\begin{proof}
$H(ab) \geq H(b|x)+H(a|y)$ by the chicken and egg inequality.\\
%\begin{align*}
%H(ab) &\geq H(abx) - H(x) + H(aby) - H(y)\\
%\Rightarrow H(ab) &\geq H(b|x) + H(a|y)\\
%\end{align*}
$H(ab) \leq H(b|x)+H(a|y)+I(x;y)$: by submodularity on $ay$ and $bx$:\\
$H(ay) + H(bx) = H(aby) + H(abx) \ge H(abxy) + H(ab)$.\\
$\Rightarrow H(ab) \le H(bx) + H(ay) - H(xy)$. 
\end{proof}

Now, we come to our case analysis for proving the forward direction of Theorem \ref{thm:2cycle}.

We first multiply our entropic vector by the least common denominator so that all the elements of the vector are integer.  We show that we can find a code that realizes this integer valued entropic vector.

Let $X_1,...X_{H(x|y)}$,$Z_1,...Z_{I(x;y)}$ be random variables originating at source x, and let $Y_1,...,Y_{H(y|x)}$, $Z_1,...,Z_{I(x;y)}$ be random variables originating at source y, where $X_i, Y_j, Z_k$ are independent for all $i,j,k$.

We split up the proof into 4 cases.  Case 1 corresponds to when $H(a)$ is greater than $H(a|x)+H(a|y)$ and $H(b)$ is greater than $H(b|x)+H(b|y)$.  Case 4 takes care of the instances when both $H(a)$ is less than $H(a|x)+H(a|y)$ and $H(b)$ is less than $H(b|x)+H(b|y)$.  Cases 2 and 3 are symmetric corresponding to when exactly one of $H(a)$ and $H(b)$ is greater than the sum of the conditional entropy on $x$ and $y$.  Cases 1,2 (or 3),4 correspond to distinctly different coding functions on edges $a$ and $b$.  Case 1 has the simplest codes - we send bits uncoded with the exception of possibly XORing $X$ and $Z$ or $Y$ and $Z$.  In cases 2 and 3 we need to XOR bits of $X,Y$ on one edge, and in case 4 we need to XOR bits of $X,Y$ on both edges in a manner similar to the example in Figure \ref{fig:offsetb}. 

\noindent \textbf{Case 1:}\\
$H(a) = H(a|x) + H(a|y) + f$, $f \ge 0$ and note $f \leq I(x;y)$ by Observation \ref{obs:boundHa}.\\
$H(b) = H(b|x) + H(b|y) + g$, $g \ge 0$ and note $g \leq I(x;y)$ by Observation \ref{obs:boundHb}.\\
$H(ab) = H(b|x) + H(a|y) + h$, and note $0 \leq h \leq I(x;y)$ by Observation \ref{obs:boundHab}.\\

\begin{obs}
$h \leq H(a|x) + H(b|y) + f + g$
\end{obs}

\begin{proof}
Implied by submodularity on $a$ and $b$.
\end{proof}

\begin{obs}
$h \ge \max(f,g)$
\label{obs:boundh}
\end{obs}

\begin{proof}
\begin{align*}
H(x|a) & \ge H(x|ab) \text{           Conditioning reduces entropy}\\
H(ax) - H(a) & \ge H(abx) - H(ab)\\
H(ab) - H(bx) - H(a|y) + H(x) & \ge H(a) - H(ax) - H(a|y) + H(x)\\
H(ab) - H(b|x) - H(a|y) & \ge H(a) - H(a|x) - H(a|y)\\
h & \ge f\\
\end{align*}

The proof that $h \ge g$ is similar.

\end{proof}

We claim that the following code realizes the entropic vector and is serializable:\\
For notational convenience let $Z'_1 = Z_{f+1}, Z'_2 = Z_{f+2}, ..., Z'_{h - f - g} = Z_{h-g}$.  Any $Z'_i$ with $i> h-f-g$ we will take to be 0.\\
$$f_a = X_1,\ldots, X_{H(a|y)}, Y_1+ Z'_{1}, \ldots, Y_{H(a|x)} + Z'_{H(a|x)}, Z_1,\ldots Z_f$$
$$f_b = X_1+Z'_{H(a|x)+1}, \ldots, X_{H(b|y)}+Z'_{H(a|x)+H(b|y)}, Y_1, \ldots, Y_{H(b|x)}, Z_{h-g-1}, \ldots, Z_h$$
This is a valid code because $H(x|y) \ge H(a|y) \ge H(b|y)$, $H(y|x) \ge H(b|x) \ge H(a|x)$, $h\leq I(x;y)$, $h \leq H(a|x) + H(b|y) + f + g$, and $h \ge \max(f,g)$. It is easy to check that this code realizes the entropic vector.  It is serializable because $H(a|y)\geq H(b|y)$, $H(b|x)\geq H(a|x)$ and both sources know $Z$.\\  
%\Anote{Is it clear what I mean by ``valid code'' or do I need to be more specific?}

\noindent \textbf{Case 2:}\\
$H(a) = H(a|x) + H(a|y) - f$, $f \ge 0$ and note $f \le \min(H(a|x),H(a|y))$ by Observation \ref{obs:boundHa}.\\
$H(b) = H(b|x) + H(b|y) + g$, $g \ge 0$ and note $g \leq I(x;y)$ by Observation \ref{obs:boundHb}.\\
$H(ab) = H(b|x) + H(a|y) + h$, and note $0 \leq h \leq I(x;y)$ by Observation \ref{obs:boundHab}.\\

\begin{obs}
$h \leq (H(a|x) - f) + H(b|y) + g$
\end{obs}

We claim that the following code realizes the entropic vector and is serializable:\\
Any $Z_i$ with $i> h$ we will take to be 0.\\
$$f_a = X_1+Y_1, X_2+Y_2, \ldots, X_f+Y_f, X_{f+1}, \dots, X_{H(a|y)}, Y_{f+1}+ Z_{g+1}, \ldots, Y_{H(a|x)} + Z_{g+H(a|x)-f}$$
$$f_b = X_1+Z_{g+H(a|x)-f+1}, \ldots, X_{H(b|y)}+Z_{g+H(a|x)-f+H(b|y)}, Y_1, \ldots, Y_{H(b|x)}, Z_{1}, \ldots, Z_g$$
This is a valid code for the same reasons as Case 1, and also because $h\leq (H(a|x) - f) + H(b|y) + g$, and $f \le H(a|x)$ and $f \le H(a|y)$. It is easy to check that this code realizes the entropic vector; here it is important that $g \le h$ which is true by the argument from Observation \ref{obs:boundh}.  It is serializable because we can send $Y_1,..., Y_{H(b|x)}$ along edge $b$, then because $H(b|x)\geq H(b|y)$ we can send everything along edge $a$, and then because $H(a|y) \ge H(a|x)$ we can send all the $X$s and $Z$s on edge $b$.\\  

\noindent \textbf{Case 3:}\\
$H(a) = H(a|x) + H(a|y) + f$, $f \ge 0$ and note $f \leq I(x;y)$ by Observation \ref{obs:boundHa}.\\
$H(b) = H(b|x) + H(b|y) - g$, $g \ge 0$ and note $g \le \min(H(b|x),H(b|y))$ by Observation \ref{obs:boundHb}.\\
$H(ab) = H(b|x) + H(a|y) + h$, and note $0 \leq h \leq I(x;y)$ by Observation \ref{obs:boundHab}.\\

Symmetric to Case 2.\\

\noindent \textbf{Case 4:}\\
$H(a) = H(a|x) + H(a|y) - f$, $f \ge 0$ and note $f \le \min(H(a|x),H(a|y))$ by Observation \ref{obs:boundHa}.\\
$H(b) = H(b|x) + H(b|y) - g$, $g \ge 0$ and note $g \le \min(H(b|x),H(b|y))$ by Observation \ref{obs:boundHb}.\\
$H(ab) = H(b|x) + H(a|y) + h$, and note $0 \leq h \leq I(x;y)$ by Observation \ref{obs:boundHab}.\\

Applying the inequality $H(a) + H(b) > H(a|x) + H(b|y)$, together with the fact that $H(a) \ge H(a|x)$ and $H(b) \ge H(b|x)$ implies that at least one of $H(a) > H(a|x)$, $H(b) > H(b|y)$ holds.  Or, written in terms of $f,g$ this means that at least one of $f < H(a|y)$, $g<H(b|x)$ holds.

\begin{obs}
$h \leq (H(a|x) - f) + (H(b|y) - g)$
\end{obs}

\noindent \textbf{Case 4a:} $f< H(a|y)$\\

We claim that the following code realizes the entropic vector and is serializable:\\
Any $Z_i$ with $i> h$ we will take to be 0.\\
$$f_a = X_2+Y_1, X_3+Y_2, \ldots, X_{f+1}+Y_f, X_1, X_{f+2}, \dots, X_{H(a|y)}, Y_{f+1}+ Z_{1}, \ldots, Y_{H(a|x)} + Z_{H(a|x)-f}$$
$$f_b = X_1+Y_1, X_2+Y_2, \ldots, X_g+Y_g, X_{g+1}+Z_{H(a|x)-f+1}, \ldots, X_{H(b|y)}+Z_{H(a|x)-f+H(b|y)-g}, Y_{g+1}, \ldots, Y_{H(b|x)}$$  
This is a valid code because $f+1 \le H(a|y)$, $h \leq (H(a|x) - f) + (H(b|y) - g)$, $f \le \min(H(a|x),H(a|y))$ and $g \le \min(H(b|x),H(b|y))$. It is easy to check that this code realizes the entropic vector.  To show it is serializable, we first consider the case when $f \leq g$: we can send $X_1$ along edge $a$; then $X_1+Y_1$ along edge $b$; then $X_2+Y_1$ along edge $a$; \ldots; then $X_{f+1}+Y_f, X_{f+2}, \dots, X_{H(a|y)}$ along edge $a$; then because $H(a|y)\geq H(b|y)$, we can send then everything along edge $b$; and then since $H(b|x) \ge H(a|x)$ we can complete the transmission for edge $a$.  The case for $f>g$ is very similar.\\  

\textbf{Case 4b:} $g< H(b|x)$\\
This case is similar, but we switch the roles of edge $a$ and edge $b$.
\section{Proofs omitted from Section~\ref{sec:serialize}}

\subsection{Linear codes}
\label{ap:linear}
\renewcommand{\algorithmiccomment}[1]{    $/^*$\textit{ #1 }$^*/$}
% \nc{\RETURN}{\textbf{return}}
\begin{algorithm}[h!]
  \caption{Greedy Algorithm for Linear Codes} \label{alg:greedy-linear}
  \CallSer
  \begin{algorithmic}[1]
  \STATE \COMMENT{$\nwc = \networkcode{}, G = (V,E,S)$ is a network code with coding functions over field $\field$.  We construct $\Sigma_e^{(1..k)}$ and $f_e^{(1..k)}$.}  
  \STATE $A_e \gets 0$ for all $e \in E$. \COMMENT{$A_e \subseteq T_e$ represents the information we have sent over edge $e$}
  \STATE $A_s \gets T_s$ for all $s \in \Src$.
	\STATE $B_e \gets T_e \intsct \left(\Spn_{s \in \In(e)} A_s\right)$ for all $e \in E$. \COMMENT{$B_e \subseteq T_e$ represents the information that the tail of $e$ knows about $T_e$}
	\STATE $i = 1$
  \WHILE {$\exists \; e=(u,v)$ in $G$ such that $A_e \neq B_e$} 
\STATE Let $x_e$ be any vector in $B_e$ that doesn't lie in $A_e$
\STATE $\Sigma_e^{(i)} \gets \field$, $f_e^{(i)} \gets x_e$
\STATE $A_e \gets A_e \spn \balg{x_e}$
\STATE $\forall \; e' \in E, e' \ne e, \Sigma_{e'}^{(i)} \gets 0$, $f_{e'}^{(i)} \gets 0$
\STATE $\forall \, e' = (v, \cdot) \in E$, $B_{e'} \gets T_{e'} \intsct (B_{e'} \spn \{x_e\})$ \COMMENT{Node $v$ ``learns'' $x_e$}
\STATE $i++$
  \ENDWHILE
  %\RETURN $\{A_e\}$, $\nwc'$
  \end{algorithmic}
\end{algorithm}

\begin{thm}[Theorem \ref{thm:TFAE} restated]
For a network code $\nwc = \networkcode{}$, the following are equivalent:
\begin{enumerate}
\item \label{item:ser} $\nwc$ is not serializable
\item \label{item:greedy} \CallSer returns $\{A_e\}$ \st $\exists \, e, A_e \ne T_e$
\item \label{item:IV} $\nwc$ has a non-trivial information vortex
\end{enumerate}
\end{thm}

\begin{proof}
\underline{$\neg$\textbf{\ref{item:greedy} $\Rightarrow$ $\neg$\ref{item:ser}}} If \CallSer returns $\{A_e\}$ \st $\forall e, \,A_e = T_e$ then $\nwc$ is serializable:\\
We show that the $f_e^{(1..k)}, \Sigma_e^{(1..k)}$ created by \CallSer satisfy the conditions in Definition \ref{def:ser}:
\begin{enumerate}
\item $f_e^{(i)}: \msg \rightarrow \Sigma_e^{(i)}$ by construction.
\item The non-zero functions $f_e^{(i)}$ form a basis for $T_e$.  Because linear maps are indifferent to the choice of basis, if $f_e(m_1) = f_e(m_2)$ then in any basis, each coordinate of $f_e(m_1)$ equals the corresponding coordinate of $f_e(m_2)$, and thus $f_e^{(i)}(m_1) = f_e^{(i)}(m_2)$ for all $i$. 
\item If $f_e(m_1) \ne f_e(m_2)$ then for any basis we choose to represent $f_e$, the values $f_e(m_1), f_e(m_2)$ will differ in at least one coordinate, and thus $\exists i, \; f_e^{(i)}(m_2) \ne f_e^{(i)}(m_2)$.
\item When we assign a function $f_e^{(i)} = x_e$ we have that $x_e$ is in $B_e$ which guarantees it is computable from information already sent to the tail of $e$.
\end{enumerate}

\noindent \underline{\textbf{\ref{item:greedy} $\Rightarrow$ \ref{item:IV}}} If \CallSer returns $\{A_e\}$ \st $\exists \, e \,A_e \ne T_e$ then $\nwc$ has a non-trivial IV.\\
We claim the the vector spaces $\{A_e\}$ returned by \CallSer form a non-trivial IV.  $\{A_e\}$ is non-trivial by hypothesis, so it remains to show it is an $IV$.
$\{A_e\}$ satisfies property (1):  For each $S \in \Src$, $A_S = T_S$ by construction (Line 3 of \CallSer).\\
$\{A_e\}$ satisfies property (2):  By induction on our algorithm, $B_e$ is exactly $T_e \intsct \left( \Spn_{e' \in \In(e)} A_{e'} \right)$.  At termination, $B_e = A_e$ for all $e \in E$.  So, we have that $A_e = T_e \intsct \left( \Spn_{e' \in \In(e)} A_{e'} \right)$.\\
 
\noindent \underline{\textbf{\ref{item:IV} $\Rightarrow$ \ref{item:ser}}} If $\nwc$ has a non-trivial IV then it isn't serializable.\\
Suppose for contradiction that $\nwc = \networkcode{}, G = (V,E,\Src)$ is serializable.  %The general idea of the proof is simple. If we have an IV in $\nwc$ and the graph is serializable then we must have an IV in the serialization DAG.  But we cannot have an IV in a DAG because there is no first edge that cannot distinguish between every message in its alphabet.  
Let $f_e^{(1..k)}$ and $\Sigma_e^{(1..k)}$ satisfy the conditions of definition \ref{def:ser}.  Let $\{W_e\}$ be a non-trivial IV for $\nwc$.

We say that a function $f_e^{(j)}$ has property $P$ if there $\exists \; m_1, m_2 \in \msg$ such that $f_e^{(j)}(m_1) \ne f_e^{(j)}(m_2)$ and $m_1, m_2 \in W_e^{\perp}$.  There must be such a function since our IV is non-trivial and $\Sigma_e^{(1..k)}, f_e^{(1..k)}$ is a serialization of $\nwc$.  Let $i^*$ be the smallest $i$ such that any function satisfies property $P$ and suppose $f_{e^*}^{i^*}$ satisfies $P$ with messages $m_1^*, m_2^*$.  

By definition, $W_{e^*} = T_{e^*} \intsct \left( \Spn_{e' \in \In(e^*)} W_{e'} \right)$,  so $m_1^*, m_2^* \in W_{e^*}^{\perp}$ implies that for all $e' \in \In(e^*)$, $m_1^*, m_2^* \in W_{e'}^{\perp}$.  But, $f_{e^*}^{i^*}$ can distinguish between $m_1^*, m_2^*$ so at least one of $e' \in \In(e^*)$ must also be able to distinguish between $m_1^*, m_2^*$ at a time \emph{before} $i^*$.  Therefore, there exists some $f_{e'}^{i'}$, $i'<i^*$ that satisfies property $P$, a contradiction to the fact that $i^*$ was the smallest such $i$.
\end{proof}

\section{Proofs omitted from Section~\ref{sec:general}}
\label{ap:general}

The following lemma is standard; for completeness,
we provide a proof here.

\begin{lem} \label{lem:balg}
Suppose $f_1:S \rightarrow T_1$
and $f_2:S \rightarrow T_2$ are two functions
on a set $S$.
\begin{enumerate}
\item \label{lem:balg-1}
$\balg{f_2} \subseteq \balg{f_1}$ if and only if 
there exists a function $g:T_1 \rightarrow T_2$
such that $f_2 = g \circ f_1.$
\item \label{lem:balg-2}
Suppose $S$ is finite.  
If $f_1 \times f_2$ denotes the function 
$S \rightarrow T_1 \times T_2$ defined by
$x \mapsto (f_1(x),f_2(x))$, then
$\balg{f_1} \spn \balg{f_2} = \balg{f_1 \times f_2}.$
\end{enumerate}
\end{lem}

\begin{proof}
A Boolean subalgebra of $2^S$ can be equivalently
described as a collection of subsets of $S$, closed
under union, intersection, and complementation,
by equating a $\{0,1\}$-valued function $b$ with
the set $b^{-1}(1).$  In this proof we adopt the 
``collection of subsets'' definition of a Boolean
subalgebra of $2^S$, since it is more convenient.
Note that under this interpretation, if $f:S
\rightarrow T$ is any function then $\balg{f}$ 
consists of all subsets of the form $f^{-1}(U), 
\; U \subseteq T$.

If $f_2 = g \circ f_1$ for some $g$, then every set
of the form $f_2^{-1}(U)$ can be expressed as 
$f_1^{-1}(g^{-1}(U))$ which shows that $\balg{f_2}
\subseteq \balg{f_1}.$  Conversely, if 
$\balg{f_2} \subseteq \balg{f_1}$ then 
for every $u \in T_2$ the set $f_2^{-1}(\{u\}) \in \balg{f_2}$
belongs to $\balg{f_1}$, i.e. it can be expressed as $f_1^{-1}(V_u)$ 
for some set $V_u \subseteq T_1.$  The sets $f_1^{-1}(V_u)$
are disjoint as $u$ ranges over the elements of $T_2$
so the sets $V_u$ themselves must be disjoint.  Define
$g(v) = u$ if $v \in V_u$ for some $u \in T_2$, and
define $g(v)$ to be an arbitrary element of $T_2$ otherwise.
For any $x \in S,$ if $u=f_2(x)$ then $x \in f_2^{-1}(u) = 
f_1^{-1}(V_u)$, which implies that $g(f_1(x)) = u.$
Hence $f_2 = g \circ f_1$ as desired.

To prove (\ref{lem:balg-2}) we argue as follows.
Clearly $\balg{f_1}, \balg{f_2} \subseteq \balg{f_1 \times f_2}$,
so $\balg{f_1} \spn \balg{f_2} \subseteq \balg{f_1 \times f_2}$ as
well.  For the reverse inclusion, note that every element
of $\balg{f_1 \times f_2}$ can be expressed as a finite
union of sets of the form $(f_1 \times f_2)^{-1}(t_1,t_2)$.
Every such set can be expressed as $f_1^{-1}(t_1) \cap
f_2^{-1}(t_2)$, which proves that it belongs to
$\balg{f_1} \spn \balg{f_2}$.
\end{proof}

\begin{algorithm}[h!]
  \caption{Greedy algorithm for general network codes} \label{alg:greedy-linear}
  \GenSer
  \begin{algorithmic}[1]
  \STATE \COMMENT{$\nwc = \networkcode{}, G = (V,E,S)$ is a network code.} 
  \STATE \COMMENT{We construct $\Sigma_e^{(1..k)}$ and $f_e^{(1..k)}$.}  
  \STATE $\bal_e \gets 0$ for all $e \in E$. \COMMENT{$\bal_e \subseteq \balg{f_e}$ represents the information we have sent over edge $e$}
  \STATE $\bal_s \gets \balg{f_s}$ for all $s \in \Src$.
	\STATE $B_e \gets \balg{f_e} \intsct \left(\Spn_{s \in \In(e)} \bal_s\right)$ for all $e \in E$. 
        \STATE \COMMENT{$B_e \subseteq \balg{f_e}$ represents the information that the tail of $e$ knows about $f_e$}
	\STATE $i = 1$
  \WHILE {$\exists \; e=(u,v)$ in $G$ such that $\bal_e \neq B_e$} 
\STATE Let $x_e$ be any binary-valued function in $B_e \setminus \bal_e$.
\label{line:xe}
\STATE $\Sigma_e^{(i)} \gets \{0,1\}$, $f_e^{(i)} \gets x_e$
\STATE $A_e \gets A_e \spn \balg{x_e}$
\STATE $\forall \; e' \in E, e' \ne e, \Sigma_{e'}^{(i)} \gets \{0\}$, $f_{e'}^{(i)} \gets 0$
\STATE $\forall \, e' = (v, \cdot) \in E$, $B_{e'} \gets \balg{f_{e'}} \intsct (B_{e'} \spn \balg{x_e})$ \COMMENT{Node $v$ ``learns'' $x_e$}
\STATE $i \gets i+1$
  \ENDWHILE
  %\RETURN $\{A_e\}$, $\nwc'$
  \end{algorithmic}
\end{algorithm}

\begin{lem} \label{lem:restriction}
For a given network code $\Phi$, restrictions 
$\Phi'$ of $\Phi$ are in one-to-one correspondence with
semi-vortices $\gdiv$.  The correspondence
maps $\gdiv$ to the network code $\Phi'[\gdiv]$
whose alphabets are $\Sigma'_e = \atoms(\bal_e)$
and whose coding functions are the functions
$f'_e = f_{\bal_e}$ defined in Definition~\ref{def:balg}.
The inverse correspondence maps $\Phi'$ to the
unique semi-vortex $\gdiv[\Phi']$ satisfying
$\bal_r = \balg{f'_r}$ for all $r \in E \cup S.$
\end{lem}
\begin{proof}
Suppose $\gdiv=(\bal_r)$ is a semi-vortex and $\Phi'[\gdiv]$
is defined as stated, with coding functions $f'_e = 
f_{\bal_e}$.  For all $e =(u,v)\in E$,
let $\In(e) = \{e_1,\ldots,e_k\}$ and let $f_i = f'_{e_i}.$ 
The relation
\[
\bal_e = \left( \spn_{i=1}^k \bal_{e_i}
\right) \intsct \balg{f_e}
\]
implies that 
\[
\bal_e \subseteq \spn_{i=1}^k \bal_{e_i} =
\spn_{i=1}^k \balg{f_i} = \balg{(f_1,\ldots,f_k)},
\]
where the last equation follows from Lemma~\ref{lem:balg}.
Since $\balg{f_{\bal_e}} = \bal_e \subseteq \balg{(f_1,\ldots,f_k)},$
we can apply Lemma~\ref{lem:balg} again to conclude that
$f_{\bal_e} = g \circ (f_1,\ldots,f_k)$ for some $g$.  Thus
$\Phi'[\gdiv]$ is a network code.  To prove that it is 
a restriction of $\Phi$, we use the containment
$\bal_e \subseteq \balg{f_e}$ for every edge $e \in E \cup S$,
together with Lemma~\ref{lem:balg},
to construct the functions $g_e : \Sigma_e \rightarrow \Sigma'_e$
required by the definition of a restriction of $\Phi.$
\end{proof}

\begin{lem} \label{lem:comparison}
If $\Phi'$ is a restriction of $\Phi$,
and $\Phi'$ is serializable, then $\gdiv[\Phi']$
is contained in every GIV of $\Phi$.
\end{lem}
\begin{proof}
Suppose $\Phi'$ is  a serializable restriction of $\Phi$,
with serialization consisting of alphabets $\Sigma_e^{(i)}$
and coding functions $f_e^{(i)}$.

Suppose now that $\gdiv=\{\bal_e\}_{e \in E \cup S}$ is any 
GIV of $\Phi$.
First, we claim $\balg{f_e^{(i)}} \subseteq \bal_{e}$ for 
every edge $e$.  To prove the claim we use 
induction on $i$.
The claim is clearly true when $i=0$.
Otherwise, let $e_1,\ldots,e_r$ be the edges in 
$\In(e)$.  By Lemma~\ref{lem:balg},
the existence of a function $h_e^{(i)}$ such that
$f_e^{(i)}(m) = h_e^{(i)} \left( \prod_{j=1}^r
f_{e_j}^{1..i-1} \right)$ implies
the first of the following containments:
\begin{equation} \label{eq:tfae-general-1}
\balg{f_e^{(i)}} \subseteq 
 \spn_{j=1}^r \spn_{\ell=1}^{i-1} \balg{f_{e_j}^{(\ell)}} 
 \subseteq \spn_{j=1}^r \bal_{e_j}.
% \subseteq  \bal_{h(u)}.
\end{equation}
The second containment in \eqref{eq:tfae-general-1} 
follows from our induction hypothesis.  
% The third
% follows from Property~\ref{gdiv-v} of a GDIV, applied
% to $v = h(u).$
Now, property~\ref{def:ser:2} of a serialization 
implies that
$\balg{f_e^{(i)}} \subseteq \balg{f'_{e}}.$ 
Combining this with \eqref{eq:tfae-general-1}
% and using $\tl(h(e)) = h(\tl(e))$, 
we obtain
\begin{equation} \label{eq:tfae-general-2}
\balg{f_e^{(i)}} \subseteq \left( \spn_{j=1}^r \bal_{e_j} \right) 
\intsct \balg{f'_{e}} 
= \bal_{e},
\end{equation}
as desired.  

% Now to show that $h$ is not an epimorphism,
% we argue by contradiction.
If $\gdiv[\Phi']$ is not contained in $\gdiv$, then
there exists an edge $e$ of $G$ 
such that 
% $\balg{f_{e'}} \ne \bal_{e'}$.   
% Property~\ref{gdiv-e} of a GDIV guarantees 
% that $\bal_{e'} \subseteq \balg{f_{e'}}$, so 
\begin{equation} \label{eq:tfae-general-3}
\balg{f'_{e}} \not\subseteq \bal_{e}.
\end{equation}
Property~\ref{def:ser:2} of a serialization implies
the existence of a function
$H : \Sigma'_e \rightarrow \prod_{i=1}^k \Sigma_e^{(i)}$
such that $H(f'_e(m)) = (f_e^{(1)}(m),\ldots,f_e^{(k)}(m))$
for all $m \in \msg.$  Property~\ref{def:ser:3} 
implies that $H$ is one-to-one, hence it has
a left inverse: a function $G :
\prod_{i=1}^k \Sigma_e^{(i)} \rightarrow \Sigma'_{e}$
such that $G \circ H$ is the identity.
Letting $F=\prod_{i=1}^k f_e^{(i)}$, the 
definition of $H$ implies that
$F = H \circ f'_{e}$, whence $f'_{e} = G \circ F$.
Applying Lemma~\ref{lem:balg} once more,
\[
\balg{f'_{e}} \subseteq \balg{F} = \spn_{i=1}^k \balg{f_e^{(i)}},
\]
and the right side is contained in $\bal_{e}$ by
\eqref{eq:tfae-general-2}.  This contradicts
\eqref{eq:tfae-general-3}, which completes the
argument.
\end{proof}

\begin{lem} \label{lem:semi-vortex}
At the start of any iteration of the main loop of 
\GenSer\, the following invariants hold.
\begin{enumerate}
\item $A_e = \balg{f_e^{(1)},\ldots,f_e^{(i-1)}}$ for all $e \in E$.
\item $B_e = \balg{f_e} \intsct
\left( \spn_{\hat{e} \in \In(u)} \bal_{\hat{e}} \right)$
for all $e \in E$.
\item The collection of subalgebras
$\gdiv=\{\bal_e\}_{e \in  E \cup S}$ constitutes a 
semi-vortex.
\item $\Phi'[\gdiv]$ 
is a serializable restriction of $\Phi$. 
\end{enumerate}
\end{lem}
\begin{proof}
% The facts that $B_e = \balg{f_e} \intsct
% \left( \spn_{\hat{e} \in \In(u)} \bal_{\hat{e}} \right)$
% for all $e$ and that $\gdiv$ is a semi-vortex
The first three invariants
can be verified by a trivial induction
on the number of loop iterations.  We claim
that $\Phi'[\gdiv]$ is serializable, and in
fact that the coding functions $\{f_e^{(j)}\}$
constructed in the preceding iterations of the
main loop constitute a serialization of $\Phi'[\gdiv].$
For property~\ref{def:ser:1} of a serialization,
there is nothing to check.  To prove property~\ref{def:ser:2},
observe that $f_e^{(j)} \in \bal_e = \balg{f'_e}$,
which implies by Lemma~\ref{lem:balg}
that $f_e^{(j)} = b \circ f'_e$ for 
some binary-valued function $b$ on $\Sigma'_e$.
If $f'_e(m_1)=f'_e(m_2)$ then $b(f'_e(m_1))=b(f'_e(m_2))$,
which establishes property~\ref{def:ser:2}.
To prove property~\ref{def:ser:3}, observe
that $\bal_e=\balg{f'_e}$ is generated by the
functions $f_e^{(1..i-1)}$, so if $f'_e(m_1) \neq f'_e(m_2)$
then there is some $j \leq i-1$ such that $f_e^{(j)}(m_1) \neq
f_e^{(j)}(m_2).$  Finally, property~\ref{def:ser:4}
follows from the structure of the algorithm itself.  Either
$f_e^{(j)}$ is the constant function $0$, in which
case there is nothing to prove, or $f_e^{(j)}$ is equal
to the function $x_e$ chosen in line~\ref{line:xe}
of the $j^{\mathrm{th}}$ loop iteration of \GenSer.  
In that case $x_e$ belonged to the 
Boolean algebra $B_e$ at the start of that
loop iteration, which means 
$$
x_e \; \in  \; \balg{f_e} \intsct \left( \spn_{\hat{e} \in \In(u)} A_{\hat{e}}
\right) \; \; \subseteq \; \; 
\spn_{\hat{e} \in \In(u)} A_{\hat{e}} \; \; = \; \; 
\spn_{\hat{e} \in \In(u)} 
\left( \spn_{1 \leq \ell < j} \balg{f_{\hat{e}}^{(\ell)}} \right)
$$
and another application of Lemma~\ref{lem:balg} implies
the existence of the function $h_e^{(j)}$ required by the
definition of serialization.
\end{proof}

\begin{lem} \label{lem:termination}
When \GenSer\ terminates,
$\gdiv = \{\bal_e\}_{e \in E \cup S}$ is a
GIV.
\end{lem}
\begin{proof}
Lemma~\ref{lem:semi-vortex} ensures that $\gdiv$ is 
a semi-vortex, and the algorithm's termination 
condition ensures that there is no edge $e$ such that
$\bal_e \neq B_e$.  In light of the fact that
$B_e = \balg{f_e} \intsct
\left( \spn_{\hat{e} \in \In(u)} \bal_{\hat{e}} \right),$
this implies that $\gdiv = \{\bal_e\}_{e \in E \cup S}$
is a GIV.
\end{proof}

\begin{thm}[Restatement of Theorem~\ref{thm:minmax}]
In the ordering of semi-vortices by containment, the
ones corresponding to partial serializations
have a maximal element and the GIV's have a minimal
element.  These maximal and minimal elements coincide, 
and they are both equal to the
semi-vortex $\gdiv=\{A_e\}_{e \in E \cup S}$ computed by \GenSer.
\end{thm}
\begin{proof}
By Lemmas~\ref{lem:semi-vortex} and~\ref{lem:termination}, 
$\gdiv$ is a GIV and $\Phi' = \Phi'[\gdiv]$ is a
serializable restriction of $\Phi.$
If $\Phi''$ is any other serializable
restriction of $\Phi$, then Lemma~\ref{lem:comparison} implies that
$\gdiv[\Phi''] \subseteq \gdiv$ because $\gdiv$ is
a GIV.  If $\ggdiv$ is any GIV,
then Lemma~\ref{lem:comparison} implies that 
$\ggdiv \supseteq \gdiv[\Phi'] = \gdiv$ becase
$\Phi'$ is a serializable restriction of $\Phi$.
\end{proof}

\begin{thm}[Restatement of Theorem~\ref{thm:tfae-general}]
For a network code $\Phi$ with finite alphabets, 
the following are equivalent.
\begin{enumerate}
\item \label{tfae-general-ser} $\Phi$ is serializable.
\item \label{tfae-general-greedy}
\GenSer\ outputs $\{\bal_e\}_{e \in E}$ \st $\forall \, e, A_e = \balg{f_e}.$
\item \label{tfae-general-gdiv} $\Phi$ has no nontrivial
GIV.
\end{enumerate}
\end{thm}
\begin{proof}
In the proof of Theorem~\ref{thm:minmax},
we saw that the subalgebras $\{\bal_e\}$
at the time \GenSer\ terminates constitute a GIV 
$\gdiv$ such that:
\begin{itemize}
\item $\gdiv$ is contained in every other GIV;
\item $\Phi'[\gdiv]$ is a serializable restriction
of $\Phi$;
\item $\gdiv$ contains $\gdiv[\Phi'']$ for every 
serializable restriction $\Phi''$ of $\Phi.$
\end{itemize}
Let $\Phi' = \Phi'[\gdiv].$ 
We now distinguish two cases.

% By Lemmas~\ref{lem:semi-vortex} and~\ref{lem:termination},
% at the time the algorithm terminates, the subalgebras
% $(\bal_r)$ define a GDIV $\gdiv$ such that 
% $\Phi'[\gdiv]$ is a serializable restriction
% of $\Phi.$  Let $\Phi' = \Phi'[\gdiv]$ and let
% $f'_e$ be the coding function of edge $e$ in 
% $\Phi'$.  We now consider two cases.

{\bf Case 1: $\Phi'$ is isomorphic to $\Phi$.}
In this case, we show that
all three equivalent conditions hold.  First, $\Phi$ is 
serializable because $\Phi'$ is.  Second, the fact
that $\Phi'$ is isomorphic to $\Phi$
means that $\balg{f_e} = \balg{f'_e} = \bal_e$ 
for every edge $e$.
Finally, 
we know that every GIV in $\Phi$ contains
$\gdiv$.  But 
$\gdiv = \gdiv[\Phi'] = \gdiv[\Phi]$,
which is the trivial GIV.  
By Definition~\ref{def:gdiv},
any GIV containing the trivial GIV is trivial.
Hence $\Phi$ contains no nontrivial GIV.

{\bf Case 2: $\Phi'$ is not isomorphic to $\Phi$.}
In this case, $\Phi'$ is a proper restriction of
$\Phi$, hence the semi-vortex 
$\gdiv = \gdiv[\Phi']$ 
constitutes a nontrivial GIV.  
Any serializable restriction $\Phi''$ of $\Phi$ 
satisfies $\gdiv[\Phi''] \subseteq \gdiv \subsetneq \gdiv[\Phi]$.
In particular this means that
$\Phi$ is not a serializable 
restriction of itself, i.e.~$\Phi$ is not serializable.
Finally, the statement that $\Phi'$ is not isomorphic to $\Phi$
means that for some
$e \in E$, $\balg{f'_e} \neq \balg{f_e}$. 
Recalling that $\balg{f'_e} = \bal_e$, this means
that \GenSer\ does not output $\{A_e\}_{e \in E}$ such that
$\forall \, e, A_e = \balg{f_e}.$
\end{proof}

\section{Analysis of the hitting set reduction}

\begin{reduc}[Hitting Set to Minimum $\sdef(\nwc)$]
Given a hitting set instance $(N,S)$ with universe $N = \{1,...,n\}$ and a family $S \subseteq 2^N$ of subsets of $N$, we let $S(i) = \{A_{i(1)}, A_{i(2)}, ... , A_{i(n_i)}\} \subseteq S$ where $A \in S(i)$ iff $i \in A$, and the ordering of $A$s is arbitrary.  We make the network coding instance $\nwc = \networkcode{}$ where $G$ is a directed, sourced graph with vertex set: $\{s\} \cup V \cup U \cup W \cup P$, where $V=\{v_1, v_2,...,v_n\}$, and similarly for $U, W$ and $P$; and source edges: $\{(\bullet, s)\} \cup \{(\bullet,w_i) | i \in N\}$ with messages $f_{(\bullet, s)} = \prod_{i \in N} \prod_{A \in S(i)} X_i^A$ and $f_{(\bullet, w_i)} = \prod_{A \in S(i)} Y_i^A$.  Where all $X,Y$ are uniform random variables over $\field_{2^\ell}$ for some sufficiently large $\ell>0$ to be chosen later.  Rather than enumerate our edge set $E$ and the coding functions on each edge, we will just specify the coding functions for each edge in $E$.  If a function $f_e$ is not specified, then $e$ is not in $E$.  We also show the network coding instance pictorially in Figure \ref{fig:reduc}. We use $\prod$ denote the $n$-fold cartesien product, so $\prod_{i=1}^{n} X_i$ is synonomous with the ordered $n$-tuple $(X_1,X_2,...X_n)$.  The coding functions are as follows:  \\
\begin{align*}
f_{(s,u_i)} &= \prod_{k = 1}^{n_i} X_i^{A_{i(k)}}, \; \forall i \in N  %\\ 
&f_{(s,v_i)} &= \prod_{j \in N: j \neq i} \prod_{k = 1}^{n_j} X_j^{A_{j(k)}} \; \forall i \in N \\ 
f_{(s,p_i)} &= \sum_{k = 2}^{n_i} X_i^{A_{i(k)}}, \; \forall i \in N %\\   
&f_{(v_i,v_j)} &= \prod_{A \in S(i) \cap S(j)} \sum_{k \in A} X_k^A, \; \forall i,j \in N \\      
f_{(w_i,v_i)} &= \prod_{k = 1}^{n_i} X_i^{A_{i(k)}}, \; \forall i \in N %\\   
&f_{(v_i,p_i)} &= \sum_{k = 1}^{n_i} X_i^{A_{i(k)}}, \; \forall i \in N \\       
f_{(p_i,w_i)} &= X_i^{A_{i(1)}}, \; \forall i \in N %\\ 
&f_{(w_i,u_i)} &= \prod_{k = 1}^{n_i} X_i^{A_{i(k)}}+Y_i^{A_{i(k)}}, \; \forall i \in N \\      
f_{(u_i,w_i)} &= \prod_{k = 1}^{n_i} X_i^{A_{i(k+1\mod n_i)}}+Y_i^{A_{i(k)}}, \; \forall i \in N \\  
\end{align*}
\label{reduc}
\end{reduc}

%The big idea is that for each set $A \in S$, we have a bit $\sum_{i \in A} X_i^A$ on all edges $(v_i,v_j)$ where $i,j \in A$, forming a clique.  Sending all the bits in this clique will correspond to covering set $A$.  The vertices $u_i,w_i$ make a ``switch''.  The codes between them can be serialized, and subsequently  $(w_i,v_i)$ when we add one bit to the code.  The bits we choose to add will correspond to a hitting set.  The messages on path $v_ip_iw_i$ can be sent only after all edges between $v_i$ and other $v \in V$ have been serialized.  It will pass the bit needed to serialize the switch for those elements not in the hitting set.\\

\begin{proof}[Proof of Part 1 of Lemma \ref{lem:reduc}]
Given a hitting set instance $(N,S)$ we create the network code $\nwc$ using the Reduction \ref{reduc}.

We show that $(N,S)$ has a hitting set of size $k$ if and only if $\sdef(\nwc) \le k$.\\  
%Adding a bit corresponds to expanding an edge function $f_e$ to $f_e \times \{0,1\}$.\\

\noindent $(\Rightarrow)$ Suppose $(N,S)$ has a hitting set of size $k$.  We show that $\sdef(\nwc) \le k$.\\

Let $C$ be a hitting set of size $k$.  Consider adding bit $X_c^{A_{c(1)}}$ for $c \in C$ to edge $(u_c,w_c)$.  This allows us to serialize all bits in the following stages, implicitly we are defining $f^{(1..k)}_e$ and $\Sigma_e^{(1..k)}$ for all $e \in E$:
\begin{enumerate}
	\item  For all $c \in C$, we can serialize all bits on edges $(w_c,u_c)$ and $(u_c,w_c)$: $w_c$ learns $X_c^{A_{c(1)}}$, so it can send bit $X_c^{A_{c(1)}}+Y_c^{A_{c(1)}}$ to $u_c$.  Now, this allows $X_c^{A_{c(2)}}+Y_c^{A_{c(1)}}$ to be sent on $(u_c,w_c)$, and we continue in this way until all bits serialized on these two edges.
	\item For all $c \in C$, send $f_{(w_c,v_c)} = \prod_{k = 1}^{n_c} X_c^{A_{c(k)}}$ on $(w_c,v_c)$. 
	\item For all $i \in N$, send $f_{(s,v_i)}\prod_{j \in N: j \neq i} \prod_{k = 1}^{n_j} X_j^{A_{j(k)}}$ on $(s, v_i)$ 
	\item For every set $A \in S$ there is an element $c \in A \cap C$ because $C$ is a hitting set.  Thus, there is a vertex in $V$, $v_c$ that knows $X_c^{A}$.  $v_c$ can therefore send bit $t(A) = \sum_{a \in A} X_a^A$ to all $v_a, a \in A, a\ne c$.  Now every $v_a$ knows $t(A)$ and can send it along $(v_a,v_{a'})$ for all $a' \in A, a' \ne a$.  This serializes all bits on edges between vertices in $V$.
	\item Now every vertex $v_i$ knows every bit $X$: it received all but $\prod_{k = 1}^{n_i} X_i^{A_{i(k)}}$ in step 3, and determined $\prod_{k = 1}^{n_i} X_i^{A_{i(k)}}$ in step 4. So, we can send $f_{(v_i,p_i)} = \sum_{k = 1}^{n_i} X_i^{A_{i(k)}}$ on edge $(v_{i}, p_{i})$.
	\item At $p_{i}$ we can add the code $\sum_{k = 2}^{n_i} X_{i}^{A_{i(k)}}$ from $(s, p_i)$ to $f_{(v_i,p_i)}$ to obtain $X_i^{A_{i(1)}}$ and send it on $(p_i, w_i)$.
	\item Now, for all $i \in N-C$, we can serialize $(w_i,u_i)$ and $(u_i,w_i)$ as we did in step 1.  
\end{enumerate}

\noindent $(\Leftarrow)$ Suppose that $\sdef(\nwc) \le k$.  We show that $(N,S)$ has a hitting set of size $k$.\\

Consider the partition of $E$ into sets $E(i) \forall i \in N$ such that $E(i) = \{(\cdot\,, \cdot_i)| i \in N\}$, that is $e \in E(i)$ if and only if the tail of $e$ is indexed by $i$.

\begin{lem}
For a set $A \in S$ suppose no bits are added to any edge in $\bigcup_{i \in A} E(i)$, then the bit $t(A) = \sum_{a \in A} X_a^A$ on $(v_{a}, v_{a'})$, $\forall a,a' \in A$ cannot be serialized.
\label{lem:helper}
\end{lem}

Lemma \ref{lem:helper} implies that for every set $A \in S$, at least one bit must be added on an edge in $\bigcup_{i \in A} E(i)$ to serialize $\nwc$.  In particular, if $\Phi'$ is a minimal serializable extension of $\Phi$ and let $C = \{i | \text{$\Phi'$ sends at least one additional than $\Phi$ on some edge in $E(i)$}\}$ then $C$ is a hitting set for $(N,S)$.  And $|C| \le \sdef(\Phi) \le k$.
\begin{proof}[Proof of Lemma \ref{lem:helper}]
Let $V(A) = \bigcup_{a \in A} v_a$.  Any edge $e$ going into any node in $V(A)$ must send $f_e$ because no bits are added on any of these edges.  In any serialization, it must be that for some $v_a \in V(A)$, bit $t(A)$ is sent on $(v_a, v_{a'})$ for some $a'$ before $v_a$ receives $t(A)$ from any $v_{a''} \in V(A)$.  Without loss of generality, suppose this vertex is $v_{i}, i \in A$.  Consider the subgraph induced by vertices indexed by $i$.  There is an information vortex on this subgraph with $W_{(s,v_i)} = T_{(s,v_i)}, W_{(s,u_i)} = T_{(s,u_i)}$, and $W_e = 0$ for all other edges in the subgraph.  To check this one simply has to verify that for edges $e$ out of $v_i$ $T_e \cap T_{(s,v_i)} = 0$, and similarly for $u_i$.  This implies that this subgraph is not serializable.  We don't add any bits to the subgraph by hypothesis, so to ``destroy'' this IV, and serialize the subgraph we need $W_{(v_j,v_i)} \spn W_{(s,v_i)}$ to have a non-zero intersection with $T_{(v_i,\cdot)}$.  But this is a contradiction to our choice of $i$.
%Vertex $v_{i}$ knows $X_a^A$ for all $a \in A, a \ne i$ from the coding function sent on edge $(s, v_{i})$.  To learn $X_{i}^A$ it must be sent via edge $(w_{i'},v_{i'})$.  Tracing back edges from $w_{i'}$ brings us to $s$ and back to $v_{i'}$.  All of the edges on these paths are in $E(i')$, and none of the functions $f_e$ on the edges in these paths allow us to send $X_{i'}^A$ on $(w_{i'},v_{i'})$.[This takes a bit of thought, may want to write another sentence of explanation here esp. in regards to not being able to serialize the ``switch'']      
\end{proof}
\end{proof}

\newcommand{\Tvec}[1]{
\left( \begin{smallmatrix}
1\\
x_{#1}^1 \\
\vdots \\
x_{#1}^{q-1} \end{smallmatrix}\right)} 

\begin{proof}[Proof of Part 2 of Lemma \ref{lem:reduc}]
Given a hitting set instance $(N,S)$ we create the network code $\nwc$ using the Reduction \ref{reduc}.  If $|\field|$ is large enough then one can show using facts from linear algebra (or matroid theory) that the subset $T = \left\{\Tvec{1}, \Tvec{2}, \ldots, \Tvec{p}\right\} \subset \field^q$ 
for $\{x_1, \ldots, x_p\}$ distinct elements in $\field$ has the property that any $q$ element subset forms a basis for $\field^q$; in other words, $T$ is a realization of the uniform matroid $U_{p,q}$ over $\field$.  Partition the $p$ vectors of $T$ into $|N|$ subsets $T_1, \ldots, T_{|N|}$ such that $T_i$ contains $z_i$ vectors, note that $\sum_{i \in N} z_i = p$ makes this valid.

%We now give two proofs, the first in terms of tensor products and the second using more elementary notation. 

We now consider $\Phi^q$.  Here, we will regard the edge alphabet for edge $e$ as a vector in $\Sigma_e \otimes \field^q$.  The tensor product space allows us to consider $q$ copies of $\Sigma_e$ on each edge $e$ without fixing a basis.  We claim that the extenstion of $\nwc^q$ in which we transmit the extra bits $\prod_{t \in T_i} X_i^{A_{i(1)}} \otimes t$ on edge $(u_i,w_i)$ for all $i \in N$ is serializable.

\begin{obs}
Transmitting $\prod_{t \in T_i} X_i^{A_{i(1)}} \otimes t$ along edge $(u_i,w_i)$ allows node $w_i$ to learn $\prod_{t \in T_i} X_i^{A_{i(k)}} \otimes t$ for all $k \in \{1 \ldots n_i\}$.
\label{obs:switch} 
\end{obs}

\begin{proof}
We saw in the proof of the forward direction of the part 1 of Lemma \ref{lem:reduc} that transmitting $X_i^{A_{i(1)}}$ along edge $(u_i,w_i)$ in $\nwc$ implies node $w_i$ can learn $X_i^{A_{i(k)}}$ for all $k \in \{1 \ldots n_i\}$.  This implies that in the $q$-fold repetition, transmitting $X_i^{A_{i(1)}} \otimes t$ along edge $(u_i,w_i)$ allows node $w_i$ to learn $X_i^{A_{i(k)}} \otimes t$ for all $k \in \{1 \ldots n_i\}$.
\end{proof}

\begin{obs}
If $\prod_{t \in T_a} X_a^{A} \otimes t$ can be transmitted along edge $(w_a, v_a)$ for all $a \in A$ then we can transmit $\sigma(A) = \sum_{a \in A} X_a^A \otimes \field^q$ on all edges $(v_a,v_{a'})$, $a, a' \in A$.   
\label{obs:basis}
\end{obs}

\begin{proof}
$\prod_{j \in N:j \ne a} \prod_{k=1}^{n_j} X_j^{A_{j(k)}} \otimes \field^q$ can be transmitted along $(s_a,v_a)$ for all $a \in A$, and that, together with $\prod_{t \in T_a} X_a^{A} \otimes t$ transmitted along $(w_a, v_a)$,  allows node $v_a$, for all $a \in A$, to compute $\alpha(a) = \prod_{t \in T_a} \sum_{a \in A} X_a^A \otimes t$.

Now, fix $i \in A$.  Send $\alpha(a)$ along $(v_a, v_i)$ for all $a \in A, a \ne i$.  Each message in the tuple $\alpha(a)$ is a linear combination of $\sum_{a \in A} X_a^A \otimes \field^q$ and is hence a legal message on all edges $(v_a, v_i)$.  After sending these messages, node $v_i$ knows  $\sum_{a \in A} X_a^A \otimes t$ for $\sum_{a \in A} z_a$ distinct vectors $t$.  Our $z_a$'s form a feasible fractional hitting set, thus $\sum_{a \in A} \frac{z_a}{q} \geq 1$, and $\sum_{a \in A} z_a \geq q$.  Any $q$-element subset of $T$ forms a basis of $\field^q$, and so $v_i$ can determine $\sigma(A)$.  We can then send $\sigma(A)$ on edges $(v_i,v_a), a \in A$ and then $(v_a, v_{a'})$ for all $a' \in A, a' \ne a$.    
\end{proof}

Observation \ref{obs:switch} and \ref{obs:basis} together imply that for all sets $A \in S$ we can send $\sigma(A)$ on the clique formed by $v_a, a \in A$.  A simple argument identitical to the last steps in the forward direction  of the proof of part 1 of Lemma \ref{lem:reduc} imply that we can serialize the rest of the coding functions.
\end{proof}

\section{Proofs omitted from Section~\ref{sec:asymptotic}}
\label{ap:asymptotic}

The proofs in this section rely on knowledge of tensor products.  We include a brief tutorial here for convenience.

\subsection{Tensor products}

If $V,W$ are any two vector spaces with bases
$\{\bvec^V_i\}_{i \in \mathcal{I}}$ and
$\{\bvec^W_j\}_{j \in \mathcal{J}}$, respectively, their tensor
product is a vector space $V \otimes W$ with a basis
indexed by $\mathcal{I} \times \mathcal{J}$.  The basis
vector corresponding to an element $(i,j)$ in the
index set will be denoted by $\bvec^V_i \otimes \bvec^W_j.$
For any two vectors $v = \sum_{i} a_i \bvec^V_i$ in $V$
and $w = \sum_j b_j \bvec^W_j$ in $W$, their \emph{tensor
product} is the vector 
\[
v \otimes w = \sum_i \sum_j a_i b_j \bvec^V_i \otimes \bvec^W_j 
\]
in $V \otimes W$.

\begin{lem} \label{lem:tensor-basis}
If $\{\vctr{x}_i\}_{i \in \mathcal{I}}$ and 
$\{\vctr{y}_j\}_{j \in \mathcal{J}}$ are bases
of $V,W,$ respectively, then 
$\{\vctr{x}_i \otimes \vctr{y}_j\}_{(i,j) \in 
\mathcal{I} \times \mathcal{J}}$ is 
a basis of $V \otimes W$.
\end{lem}
\begin{proof}  
It suffices to prove the lemma when $\vctr{y}_j =
\bvec^W_j$ for all $j \in \mathcal{J}.$  If the
lemma holds in this case, then by symmetry it also
holds when $\vctr{x}_i=\bvec^V_i$ for all $i \in \mathcal{I}$,
and then the general case of the lemma follows by 
applying these two special cases in succession: first
changing the basis of $V$, then changing the basis of
$W$.

To prove that $B=\{\vctr{x}_i \otimes \bvec^W_j\}_{(i,j) \in
\mathcal{I} \times \mathcal{J}}$ is a basis of 
$V \otimes W$, it suffices to prove that every
vector of the form $\bvec^V_i \otimes \bvec^W_j$
can be written as a linear combination of elements
of $B$.  By the assumption that $\{\vctr{x}_i\}_{i \in
\mathcal{I}}$ is a basis of $V$, we know that 
$\bvec^V_i = \sum_{k \in K} a_k \vctr{x}_k$ for some
finite subset $K \subseteq \mathcal{I}$ and 
scalars $(a_k)_{k \in K}.$  Now it follows that 
$\bvec^V_i \otimes \bvec^W_j = 
\sum_{k \in K} a_k \left( \vctr{x}_k \otimes
\bvec^W_j \right),$ as desired.
\end{proof}

\subsection{Basis and Rank}
\begin{defn}  
If $V$ is a vector space with basis $B$,
and $\mathcal{W} = \{W_i\}_{i \in \mathcal{I}}$
is a collection of linear subspaces, we say that 
$\mathcal{W}$ is \emph{$B$-compatible} if 
$W_i \cap B$ is a basis of $W_i$, for all $i \in
\mathcal{I}.$  We say that $\mathcal{W}$ is
\emph{basis-compatible} if there exists a basis $B$
for $V$ such that $\mathcal{W}$ is $B$-compatible.
\end{defn}

\begin{lem} \label{lem:b-compat}
If $V$ is a vector space and $\mathcal{W}$ is a 
basis-compatible collection of linear subspaces,
then $\mathcal{W}$ can be enlarged to a 
basis-compatible collection
of linear subspaces $\bar{\mathcal{W}}$ that
forms a Boolean algebra under $\spn$ and 
$\intsct$.  In particular, any three subspaces
$X,Y,Z \in \mathcal{W}$ satisfy:
\begin{align*}
(X \intsct Y) \spn Z &= (X \spn Z) \intsct (Y \spn Z) \\
(X \spn Y) \intsct Z &= (X \intsct Z) \spn (Y \intsct Z).
\end{align*}
\end{lem}
\begin{proof}
Simply let $\bar{\mathcal{W}}$ be the set of all
linear subspaces of $W \subseteq V$ such that 
$W \intsct B$ is a basis of $W$.
\end{proof}

\begin{lem} \label{lem:b-compat-triv}
If $V$ is a vector space and $\mathcal{W}$ is a
basis-compatible collection of linear subspaces,
then $\mathcal{W} \cup \{V\} \cup \{0\}$ is 
also a basis-compatible collection of linear 
subspaces.
\end{lem}
\begin{proof}
The proof is a trivial consequence of the 
definition of basis-compatible.
\end{proof}

\begin{lem} \label{lem:2-subspace}
If $V$ is a vector space and $X,Y$ are any two linear subspaces,
then $\mathcal{W} = \{X,Y\}$ is basis-compatible.
\end{lem}
\begin{proof}
Let $B_{XY}$ be any basis of $X \intsct Y$, let 
$B_X$ be any basis of $X$ containing $B_{XY}$, and
let $B_Y$ be any basis of $Y$ containing $B_{XY}$.
All the vectors in $B_X \cup B_Y$ are linearly 
independent, because if $v$ is any vector that can
be expressed as a linear combination
of elements of $B_Y \setminus B_X$ and as 
a linear combination of elements of $B_X$, 
then $v$ must belong to both $X$ and $Y$,
hence $v \in X \intsct Y$.  But the only 
element of $X \intsct Y$ that can be expressed
as a linear combination of elements of $B_Y \setminus
B_X$ is the zero vector, because $B_{XY}$ is disjoint
from $B_Y \setminus B_X$, and these two sets together 
constitute a basis of $Y$.  Hence $B_X \cup B_Y$
can be extended to a basis $B$ of $V$, and 
then $\mathcal{W} = \{X,Y\}$ is $B$-compatible.
\end{proof}

If $A,B$ are subspaces of vector spaces $V,W$,
respectively, then $A \otimes B$ is defined to
be the linear subspace of $V \otimes W$ 
consisting of all linear combinations of 
vectors in the set $\{a \otimes b \,:\, 
a \in A, \, b \in B\}.$  If $\mathcal{V}$ is a
collection of linear subspaces of $V$ and
$\mathcal{W}$ is a collection of linear subspaces
of $W$, then $\mathcal{V} \otimes \mathcal{W}$
denotes the collection of all linear subspaces
$A \otimes B \subseteq V \otimes W$ such that
$A \in \mathcal{V}$ and $B \in \mathcal{W}$.

\begin{lem} \label{lem:b-compat-product}
If $\mathcal{V}$ is a basis-compatible collection
of linear subspaces of $V$ and $\mathcal{W}$
is a basis-compatible collection of linear 
subspaces of $W$ then $\mathcal{V} \otimes
\mathcal{W}$ is a basis-compatible collection
of linear subspaces of $V \otimes W$.
\end{lem}
\begin{proof}
If $B,B'$ are bases of $V,W$, respectively, such
that $\mathcal{V}$ is $B$-compatible and 
$\mathcal{W}$ is $B'$-compatible, then 
$\mathcal{V} \otimes \mathcal{W}$ is
$(B \times B')$-compatible.
\end{proof}

\begin{coro} \label{lem:3-subspace}
If $X,Y$ are subspaces of a vector space $V$ and
$Z$ is a subspace of another vector space $W$,
then
\[
\left[ (X \otimes W) \spn (Y \otimes W) \right]
\intsct (V \otimes Z)
=
\left[ (X \otimes W) \spn (V \otimes Z) \right] 
\intsct
\left[ (Y \otimes W) \spn (V \otimes Z) \right]
.
\]
\end{coro}
\begin{proof}
By Lemmas~\ref{lem:b-compat-triv} and~\ref{lem:2-subspace}, 
we know that 
$\mathcal{V}=\{X,Y,V\}$ is basis-compatible in $V$ and
$\mathcal{W}=\{Z,W\}$ is basis-compatible in $W$.  
Hence $\mathcal{V} \otimes \mathcal{W}$ is 
basis-compatible in $V \otimes W.$  The corollary
now follows by applying Lemma~\ref{lem:b-compat}.
\end{proof}

\begin{defn}  If $V,W$ are any vector spaces,
a \emph{rank-one} element of $V \otimes W$ is an element
that can be expressed in the form $v \otimes w$
for some $v \in V, \, w \in W.$  The \emph{rank} of an
element $x \in V \otimes W$ is the minimum value of $r$
such that $x$ can be expressed as a linear combination 
of $r$ rank-one elements of $V \otimes W.$  (If $x=0$
then its rank is defined to be $0$.)
\end{defn}

\begin{lem} \label{lem:rank}
If $V,W$ are finite-dimensional vector spaces and
$x \in V \otimes W$ then the rank of $x$ is bounded
above by $\min \{\dim(V), \dim(W)\}.$
\end{lem}
\begin{proof}
Let $n=\dim(V), m = \dim(W).$  
We will assume without loss of generality that $n \leq m$
and prove that the rank of $x$ is at most $n$.
Let $\{\bvec^V_i\}$ and $\{\bvec^W_j\}$ be bases of $V,W,$
respectively.  We may express $x$ as a linear combination
\[
x = \sum_{i=1}^n \sum_{j=1}^m a_{ij} \bvec^V_i \otimes \bvec^W_j.
\]
For $1 \leq i \leq n$ let $y_i = \sum_{j=1}^m a_{ij} \bvec^W_j.$
Then
\[
x = \sum_{i=1}^n \bvec^V_i \otimes y_i,
\]
and this expresses $x$ as a sum of $n$ rank-one elements,
implying that the rank of $x$ is at most $n$.
\end{proof}

\begin{lem} \label{lem:extension}
If $V,W$ are finite-dimensional vector spaces of
dimension $n,m$, and
$P \subseteq V \otimes W$ is a linear subspace of
dimension $p$, then there exists a $pn$-dimensional
linear subspace $Q \subseteq W$ such that 
$P \subseteq V \otimes Q.$
\end{lem}
\begin{proof}
Let $\{x_1,\ldots,x_p\}$ be a basis of $P$.  By
Lemma~\ref{lem:rank}, we can write each $x_k$
in the form
\[
x_k = \sum_{i=1}^n v_{ki} \otimes w_{ki},
\]
for some vectors $v_{ki} (1 \leq k \leq p, 
\, 1 \leq i \leq n)$ in $V$ and 
$w_{ki} (1 \leq k \leq p, \, 1 \leq i \leq n)$
in $W$.  The vectors $\{w_{ki}\}$ span a 
subspace of $W$ of dimension at most $pn$.
Taking $Q$ to be any $pn$-dimensional 
subspace of $W$ containing $\{w_{ki}\}$,
we see that $P \subseteq V \otimes Q$ as
desired.
\end{proof}

\subsection{Non-serializability}

\begin{thm}[Theorem \ref{thm:asymSer} restated]
For a linear network code $\nwc = \networkcode{}$ over a field $\field$, then $\sdef(\nwc^n) \geq cn$ where $c$ is a constant dependent on $\nwc$. 
\end{thm}
\begin{proof}

If a linear network code has an information vortex
$\{W_e\}_{e \in E}$ then $\{W_e^n\}_{e \in E}$ constitutes
an information vortex in the product of $n$ copies
of the network code.  Using the fact that 
$V^n = V \otimes \field^n$ for every vector space
$V$, we may rewrite the information vortex 
as $\{W_e \otimes \field^n\}.$

Now suppose that $T$ is a 
linear subspace of $\msg^* \otimes \field^n$ 
of dimension $p$.  
(Think of $T$ as a set of bits that are added to $\nwc$ to get an  extention of $\phi$.)
Using Lemma~\ref{lem:extension}, there is a 
subspace $Q \subseteq \field^n$ of dimension 
$pd$ (where $d = \dim(\msg^*)$) 
such that $T \subseteq \msg^* \otimes Q.$
Define a new family of subspaces 
\[
X_e = (W_e \otimes \field^n) \spn (\msg^* \otimes Q).
\]
These subspaces constitute an information
vortex.  To verify this, we must check that
\begin{equation} \label{eq:must-check}
(W_e \otimes \field^n) \spn (\msg^* \otimes Q) =
\left[ (T_e \otimes \field^n) \spn (\msg^* \otimes Q) \right]
\intsct
\left[ \spn_{e' \in \mathrm{In}(e)} \left(
(W_{e'} \otimes \field^n) \spn (\msg^* \otimes Q) \right) \right].
\end{equation}
Because $\{W_e \otimes \field^n\}$ is an information
vortex, we have
\[
W_e \otimes \field^n =
(T_e \otimes \field^n) \intsct
\left(
\spn_{e' \in \mathrm{In}(e)} W_{e'} \otimes \field^n
\right).
\]
Equation \eqref{eq:must-check} now follows by applying
Corollary~\ref{lem:3-subspace} with
$X = T_e,
 Y = \spn_{e' \in \mathrm{In}(e)} W_{e'},
 Z = Q.
$
Thus the collection of subspaces $\{X_e\}$ constitutes an information vortex in the setting where $T$ added to every edge of the network code.  If the network
code is serializable in this setting, then $\{X_e\}$ must
be a trivial information vortex, implying that 
$X_e = T_e \otimes \field^n$ for every edge $e$.  
Because $\{W_e\}$ is
non-trivial, we know there is at least one edge $e$
such that $\dim(W_e) < \dim(T_e).$  
Let $m = \dim(T_e), \, k = \dim(W_e).$  
The dimension of $W_e \otimes \field^n$ is $kn.$
The dimension of $\msg^* \otimes Q$ is $p d^2.$
Hence the dimension of $X_e$ is bounded above by
$kn + pd^2.$  On the other hand, if the information
vortex $\{X_e\}$ is trivial, then $X_e = T_e \otimes \field^n$ 
implying that $\dim(X_e) = mn.$  Thus 
$$\sdef(\Phi) \ge p \geq \frac{m-k}{d^2} \cdot n.$$

\end{proof}

\end{document}